\documentclass[11pt,letterpaper]{article}
\usepackage[margin=1in]{geometry}
\usepackage[utf8]{inputenc}
\usepackage{amsmath,amssymb,amsthm,mathrsfs}
\usepackage{physics}
\usepackage{complexity}
\usepackage[ruled,linesnumbered,noend]{algorithm2e}
\usepackage{setspace}
\usepackage[noend]{algpseudocode}
\usepackage[colorlinks=true]{hyperref}
\usepackage[dvipsnames]{xcolor}
\definecolor{darkred}  {rgb}{0.5,0,0}
\definecolor{darkblue} {rgb}{0,0,0.5}
\definecolor{darkgreen}{rgb}{0,0.5,0}
\hypersetup{
  urlcolor   = blue,         
  linkcolor  = darkblue,     
  citecolor  = darkgreen,    
  filecolor   = darkred       
}
\usepackage[nameinlink, capitalize]{cleveref}
\usepackage{tcolorbox}
\usepackage{subcaption}
\usepackage{tikz}
\usetikzlibrary{decorations.pathreplacing}
\pgfmathsetmacro\MathAxis{height("$\vcenter{}$")}

\renewcommand{\poly}{\mathrm{poly}}
\usepackage[backend=biber,style=alphabetic,url=false,isbn=false,maxnames=10,minnames=3,maxalphanames=4,minalphanames=3]{biblatex}
\usepackage[toc,page]{appendix}

\newtheorem{theorem}{Theorem}[section]
\newtheorem{corollary}[theorem]{Corollary}
\newtheorem{definition}{Definition}[section]
\newtheorem{lemma}[theorem]{Lemma}
\newtheorem{fact}{Fact}[section]
\newtheorem{claim}[theorem]{Claim}
\usepackage{enumerate}
\newcommand{\atg}{\mathrm{ATG}}
\newcommand{\ghz}{\mathrm{GHZ}}
\newcommand{\under}[2]{\underbrace{#1}_{\substack{#2}}}
\newcommand{\rec}{\mathsf{Rec}}
\newcommand{\rep}{\mathsf{Rep}}
\newcommand{\pframe}{\mathsf{Frame}}
\newcommand{\res}{\mathsf{Res}}
\newcommand{\cnot}{\mathsf{CNOT}}

\begin{document}
\title{On fault tolerant single-shot logical state preparation \\ and robust long-range entanglement}

\author{Thiago Bergamaschi\thanks{UC Berkeley. \href{mailto:thiagob@berkeley.edu}{thiagob@berkeley.edu}}
\and
Yunchao Liu\thanks{UC Berkeley and Harvard University. \href{mailto:yunchaoliu@berkeley.edu}{yunchaoliu@berkeley.edu}}
}

\date{}

\maketitle

\begin{abstract}


Preparing encoded logical states is the first step in a fault-tolerant quantum computation. Standard approaches based on concatenation or repeated measurement incur a significant time overhead. The Raussendorf-Bravyi-Harrington cluster state~\cite{Raussendorf2005Longrange} offers an alternative: a single-shot preparation of encoded states of the surface code, by means of a constant depth quantum circuit, followed by a single round of measurement and classical feedforward~\cite{Bravyi2020Quantum}. In this work we generalize this approach and prove that single-shot logical state preparation can be achieved for arbitrary quantum LDPC codes. Our proof relies on a minimum-weight decoder and is based on a generalization of Gottesman's clustering-of-errors argument~\cite{Gottesman2013FaulttolerantQC}. As an application, we also prove single-shot preparation of the encoded GHZ state in arbitrary quantum LDPC codes. This shows that adaptive noisy constant depth quantum circuits are capable of generating generic robust long-range entanglement.
\end{abstract}

\section{Introduction}

Single-shot logical state preparation is a procedure in which a code state (such as the logical $\ket{0}$ or $\ket{+}$ state) of a quantum error correcting code is prepared by a constant depth quantum circuit, followed by one round of single qubit measurements, and an adaptive  Pauli correction using classical feedforward. In this paper, we prove that every CSS quantum LDPC code admits a single-shot state preparation procedure with logarithmic space overhead that is fault tolerant against local stochastic noise, generalizing prior work by \cite{Bravyi2020Quantum} for the surface code. 

The fact that single-shot logical state preparation is feasible is non-trivial: after all, logical states of quantum LDPC codes have circuit lower bounds and cannot be prepared by a constant depth quantum circuit~\cite{Bravyi2006Lieb-Robinson,Haah2016Invariant,aharonov2018quantum}. This issue is circumvented via the additional measurement and classical feedforward. Due to the simplicity of the quantum circuit, there has been a recent theoretical and experimental interest in preparing physically interesting quantum states using these adaptive shallow circuits~\cite{Lu2022Measurement,lee2022decoding,lee2024symmetry,Bravyi2022Adaptive,Tantivasadakarn2023Hierarchy,chen2023realizing,Zhu2023Nishimori,Iqbal2024Topological,Tantivasadakarn2024LongRange}. Fault tolerance is a key aspect of a faithful experimental preparation of these states, where one must argue that the state is still interesting even though the shallow circuit is noisy. Our result provides a general recipe for fault tolerant state preparation, which can be used for state initialization in fault tolerant quantum computation, as well as a candidate for an experimental demonstration of long-range entangled states using adaptive shallow circuits that is robust to noise.

The intuition behind our result is to realize \emph{repeated measurement} via \emph{measurement-based quantum computation}. The standard approach to initialize a quantum LDPC code in the logical $\ket{+}$ state is as follows: first initialize all physical qubits in a code block in the $\ket{+}$ state, then perform repeated $Z$ and $X$ syndrome measurements, and finally perform a Pauli error correction using classical feedforward based on the syndrome measurement outcomes. A key observation here is that the repeated measurements are non-adaptive, thus we can attempt to simulate this process using a cluster state, achieving a space-time trade-off. 

This naturally leads to a construction which we refer to as the Alternating Tanner Graph state $\ket{\atg}$ (\cref{fig:atg}, also known as a foliated quantum code~\cite{Bolt2016Foliated}): there are $2T+1$ copies of the code block. A copy of the $Z$ Tanner graph is placed on each odd layer (simulating $Z$ syndrome measurements), and a copy of the $X$ Tanner graph is placed on each even layer (simulating $X$ syndrome measurements). Vertical connections are added between code qubits of neighboring layers, which is used to propagate quantum computation in the time direction. To prepare this cluster state in constant depth, we initialize all qubits in the $\ket{+}$ state and apply a $\mathrm{CZ}$ gate on each edge. When applied to the surface code, the $\ket{\atg}$ is equivalent to the Raussendorf-Bravyi-Harrington (RBH) cluster state~\cite{Raussendorf2005Longrange}.

\begin{figure}[btp]

\begin{subfigure}[b]{0.5\textwidth}
\centering
    \includegraphics[width = 0.65\linewidth]{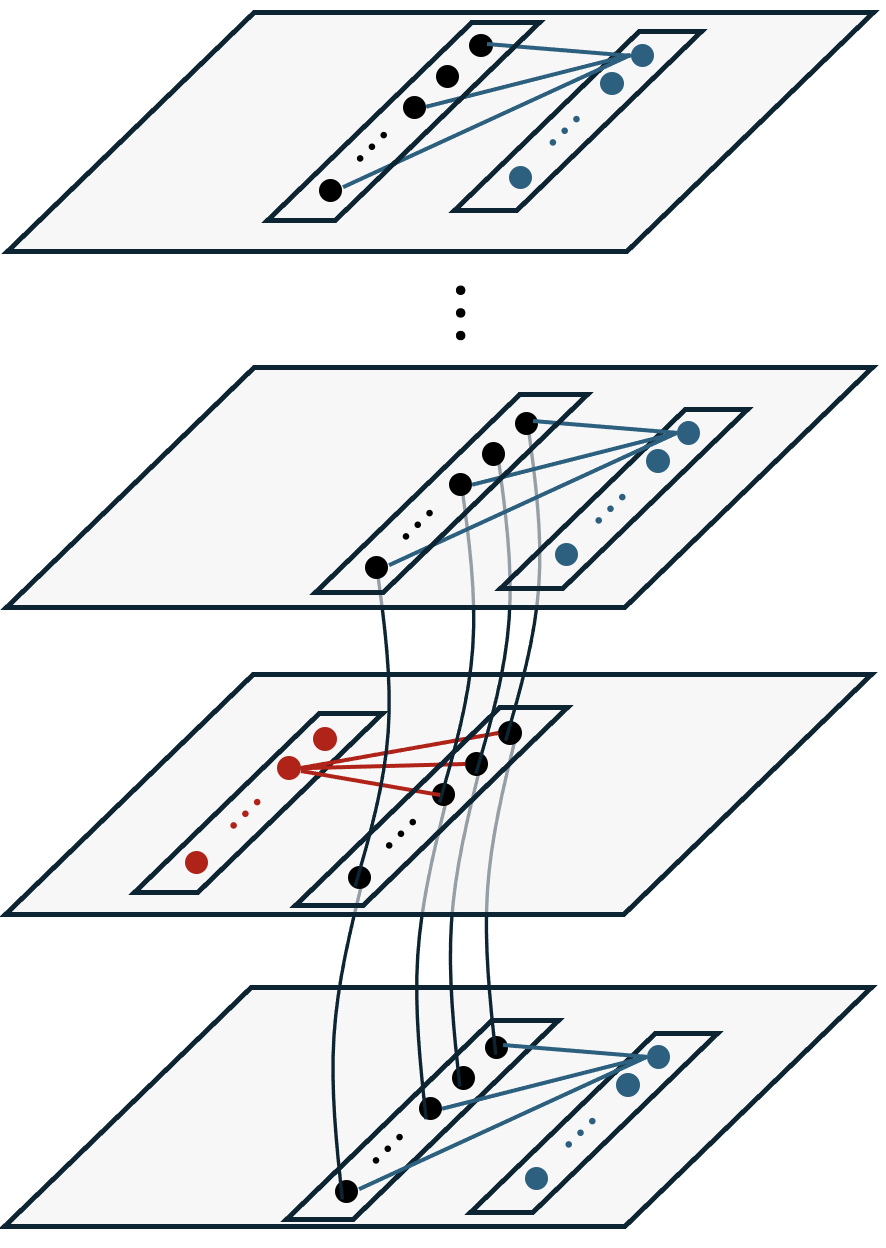}
    \caption{The Alternating Tanner Graph state $\ket{\atg}$}
\end{subfigure}
\begin{subfigure}[b]{0.5\textwidth}
\centering
    \includegraphics[width = 0.65\linewidth]{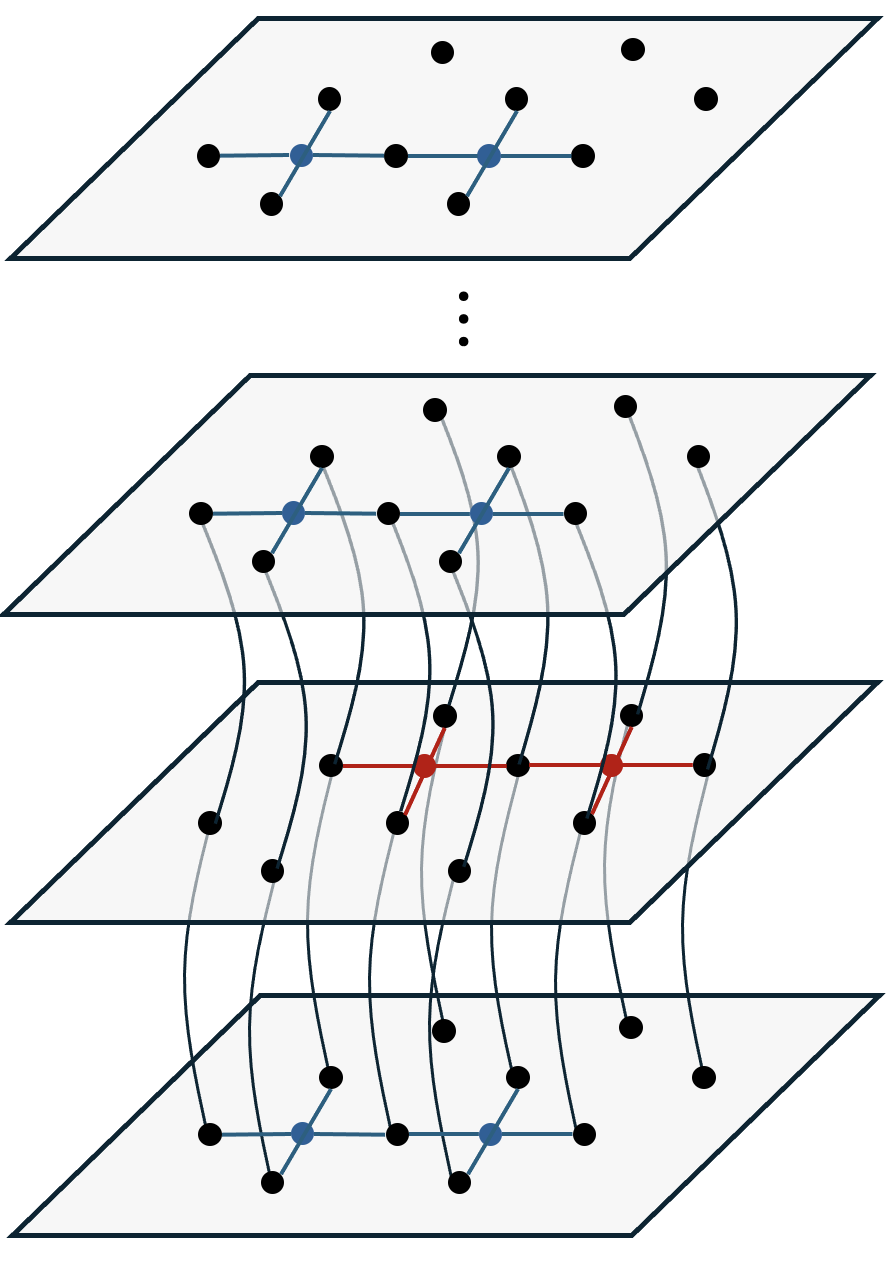}
    \caption{$\ket{\atg}$ of the surface code}
\end{subfigure}
\caption{(a) The ATG is a graph state, defined on a vertex set comprised of layers of copies of the code block and the $Z$ (resp. $X$) Tanner graph (in blue, resp. red) of the qLDPC code. (b) The RBH cluster state~\cite{Raussendorf2005Longrange}, which prepares code states of the surface code in a 3D cubic lattice arrangement, is a special case of this construction.} 
\label{fig:atg}
\end{figure}

\subsection{Results}
We start with an informal description of single-shot logical state preparation. Consider a set of noisy qubits divided into two subsystems $A$ and $B$. Single-shot logical state preparation is performed via the following procedure:
\begin{enumerate}
    \item Initialize all qubits in the (noisy) $\ket{+}$ state.
    \item Perform a (noisy) constant depth quantum circuit $W$.
    \item Measure all qubits in $A$ in the Hadamard ($\ket{+},\ket{-}$) basis, obtaining a bitstring $s\in\{0,1\}^{|A|}$. In the absence of noise, the unnormalized post-measurement state equals $\ketbra{\pm_s}_A\cdot W\ket{+}_{AB}$.
    \item Use a classical computer to calculate a Pauli operator $P(s)$ (supported on $B$) based on $s$, and apply $P(s)$ to the post-measurement state.\footnote{Classical computation is assumed to be performed instantly, such that there is no additional noise on the quantum state during this step of the computation.} 
\end{enumerate}
In the absence of noise, the final (unnormalized) state is given by
\begin{equation}
    \bigg(\ketbra{\pm_s}_A\otimes P(s)_B\bigg)\cdot W\ket{+}_{AB}.
\end{equation}

Let us begin by understanding what this procedure can achieve in the absence of noise. With a suitable choice of $W$, one can imagine that step 3 is measuring the stabilizers/parity checks of a quantum LDPC code (where $B$ contains code qubits and $A$ contains ancilla qubits for syndrome measurement), which collapses the state into a random syndrome subspace $s$ (some eigenspace of the code Hamiltonian). Step 4 then performs syndrome decoding and corrects the Pauli error, thus preparing an encoded code state.\footnote{What code state are we preparing? For example, assume there is only one logical qubit. Notice that the logical $X$ operator of the CSS LDPC code is stabilized by the initial state and commutes with the measurement. Therefore we are able to prepare the encoded $\ket{+}$ state. The encoded $\ket{0}$ state can be prepared similarly by initializing all qubits in $\ket{0}$.}

The key question is what can the above procedure achieve in the noisy setting, where every initial qubit, gate and measurement is subject to some small constant amount of noise (throughout this section, we assume the noise rate is below some constant threshold). Ideally, we would like to achieve \emph{fault-tolerant} single-shot logical state preparation, which means that the final state in the presence of noise is the desired encoded logical state, up to some residual stochastic noise. This residual noise is benign in the sense that it is correctable with the LDPC code; and the final state (with residual noise) can be used as the initial state of a fault tolerant quantum computation.



The conceptual challenge to achieve fault tolerance here is that the noisy quantum circuit only has constant depth, therefore we cannot hope to correct errors during the state preparation procedure using standard means. The RBH cluster state and its generalization, the ATG \cite{Raussendorf2005Longrange, Bolt2016Foliated}, is a carefully designed construction where after the bulk of the state is measured, the remaining boundaries/surfaces encode a logical state in the quantum error correcting code up to a Pauli correction. The intuition that this scheme might be fault tolerant lies in the redundancy embedded into the bulk ancilla qubits. Here, we give a general fault-tolerance proof that works for arbitrary quantum LDPC codes. 

\begin{theorem}\label{thm:stateprepinformal} Fix an integer $T$. There exists a fault tolerant, single shot logical state preparation procedure for the encoded $\ket{\overline{0}}, \ket{\overline{+}}$ states\footnote{Note that there are $k$ logical qubits, and we can prepare either $\ket{\overline{0}^{\otimes k}}$ or $\ket{\overline{+}^{\otimes k}}$.} of an arbitrary $[[n,k,d]]$ CSS LDPC code, using $O(n\cdot T)$ ancilla qubits, with success probability at least $1-nT\cdot 2^{-\Omega(\min(d, T))}$.
\end{theorem}

As discussed earlier, the Alternating Tanner Graph (ATG) state (\cref{fig:atg}) naturally gives rise to redundancies akin to those in a repeated syndrome measurements protocol. Our proof of fault-tolerance leverages an information-theoretic minimum-weight decoder, and is based on the clustering-of-errors argument by \cite{Gottesman2013FaulttolerantQC, Kovalev2012FaultTO} for decoding LDPC codes from repeated measurements (see below for a more detailed overview).




\paragraph{Efficient decoding.} In our single-shot logical state preparation result discussed above, the classical decoding procedure (computing $P(s)$ from $s$) relies on an information-theoretic minimum weight decoder, and is not necessarily computationally efficient.\footnote{By dynamic programming, the decoding algorithm can be made to run in $\leq T\cdot \exp(n)$ time.} A natural question is whether an efficient decoder that runs in $\poly(n)$ time exists for single-shot logical state preparation, especially if the code is known to have an efficient decoder for quantum error correction.\footnote{We emphasize the distinction between a decoder for error-correction, and a decoder for state preparation. } Here, we complement our results by showing a
simple reduction to the efficient decoding of logical state preparation based on repeated syndrome measurements.

\begin{theorem}\label{thm:stateprepmbqcinformal}
    If a family of $[[n, k, d]]$ CSS LDPC codes admits an efficient decoder for fault-tolerant logical state preparation based on $T$ rounds of repeated syndrome measurements, then it further admits an efficient decoder for a fault-tolerant single shot logical state preparation procedure using $O(n\cdot T)$ ancilla qubits. 
\end{theorem}

In other words, one can tradeoff time for space in logical state preparation. As discussed earlier, this reduction stems from mapping the repeated measurements circuit to a constant depth circuit via measurement based quantum computation, and as we illustrate naturally gives rise to the alternating tanner graph state $\ket{\atg}$.

\begin{figure}[t]
    \centering
    \begin{subfigure}[b]{0.49\textwidth}
\centering
\includegraphics[width=0.9\linewidth]{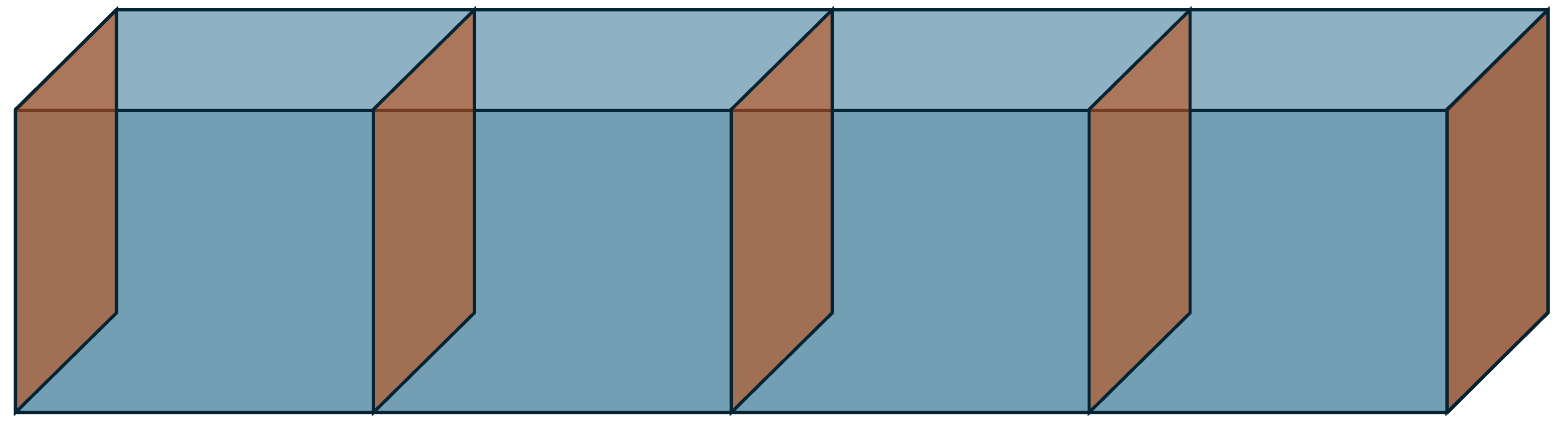}
    \caption{Measurement pattern of $\ket{\atg}$}
    \label{fig:ghz1}
\end{subfigure}
\begin{subfigure}[b]{0.49\textwidth}
\centering
\includegraphics[width=0.9\linewidth]{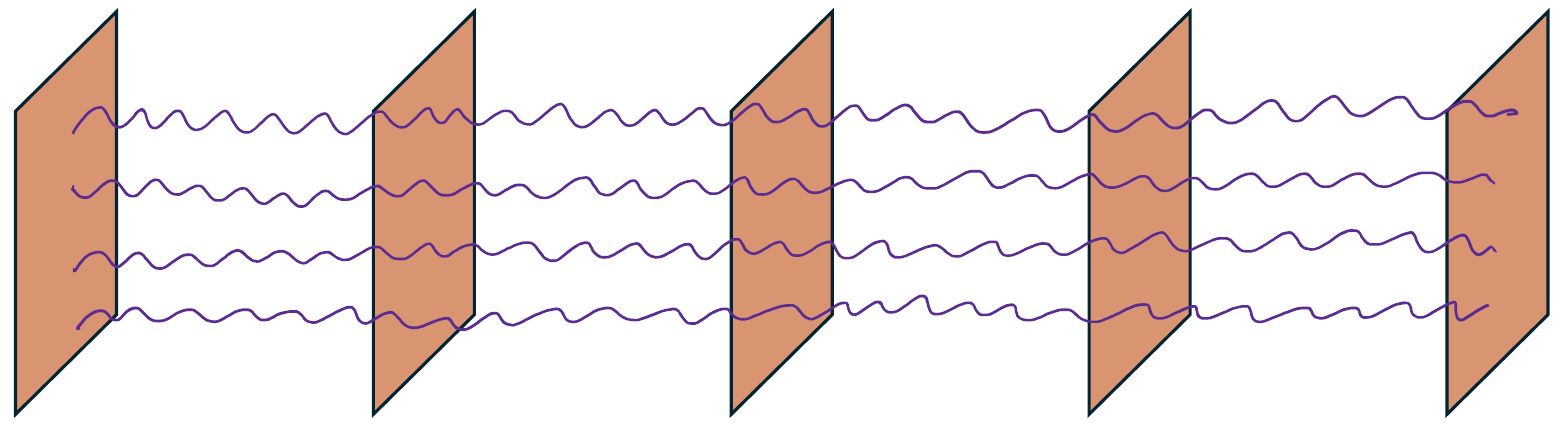}
    \caption{Encoded GHZ state}
    \label{fig:ghz2}
\end{subfigure}
    \caption{(a) We prepare $\ket{\atg}$ of an arbitrary LDPC code using a noisy shallow circuit and measure the bulk (blue), leaving $m$ surfaces (orange) unmeasured. (b) After decoding and feedforward, the state becomes the encoded $\ghz$ state across the $m$ surfaces, up to residual stochastic noise.}
    \label{fig:ghz}
\end{figure}

\paragraph{Robust long-range entanglement.} Next we show that the ATG provides a powerful framework that goes beyond the preparation of encoded $\ket{\overline{0}}$ and $\ket{\overline{+}}$ as in logical state preparation based on repeated measurements, due to the flexibility in \emph{choosing different measurement patterns} in the ATG. As an example, we show fault tolerant single-shot preparation of the encoded GHZ state 
\begin{equation}
    \ket{\ghz_m} = \frac{1}{\sqrt{2}}\left(\ket{0}^{\otimes m}+\ket{1}^{\otimes m}\right)
\end{equation}
for any $m\geq 2$ in an arbitrary quantum LDPC code. Here we summarize the algorithm for GHZ state preparation, and then discuss this result from the perspective of robust long-range entanglement and noisy shallow circuits.

To prepare the GHZ state, we measure the ATG in the following pattern, depicted (horizontally) in \cref{fig:ghz1}: we measure all the check qubits of all code blocks (blue and red in \cref{fig:atg}), as well as all code qubits in the bulk (blue in \cref{fig:ghz1}), with the exception of a careful choice of $m$ layers of $n$ code qubits (in orange). We prove that the resulting state, after minimum-weight decoding and feedforward, is in fact the encoded GHZ state across the $m$ blocks (\cref{fig:ghz2}) up to residual stochastic noise.


\begin{theorem}\label{thm:ghzinformal}
    Fix integers $m, T$. There exists a fault tolerant, single shot logical state preparation procedure for the encoded $\ket{\overline{\ghz^{\otimes k}_m}}$ state of an arbitrary $[[n,k,d]]$ CSS LDPC code, using $O(n\cdot T)$ ancilla qubits, with success probability at least $1-nT\cdot 2^{-\Omega(\min(d, T/m))}$.
\end{theorem}

This generalizes the RBH construction~\cite{Raussendorf2005Longrange} for the surface code with $m=2$ and $k=1$, that is, the encoded Bell state $\frac{1}{\sqrt{2}}\left(\ket{\overline{00}}+\ket{\overline{11}}\right)$ in the surface code. Let us review the features of the entanglement in the resulting GHZ state:
\begin{itemize}
    \item \textbf{Noise robust.} The preparation procedure uses a noisy constant depth circuit (with noise rate below a constant threshold). Yet, the resulting GHZ state is encoded in a LDPC code and is subject to residual local stochastic noise that is correctable with the code. 
    \item \textbf{Generic.} The construction works for arbitrary quantum LDPC codes.
    \item \textbf{Macroscopic.} The final state encodes $m\cdot k$ logical qubits into $m\cdot n$ physical qubits, which can achieve a constant encoding rate using high-rate codes such as hypergraph product codes~\cite{Tillich2014Quantum,Kovalev2012Improved}. 
    \item \textbf{Long range.} There are two aspects of long range entanglement: the first is the entanglement among physical qubits in each code block (which can exhibit topological order), the second is the logical GHZ entanglement that spans $m$ code blocks.  
\end{itemize}

Finally, we discuss this result from the perspective of noisy shallow quantum circuits. The first step of this state preparation algorithm is to apply a constant depth circuit on a volume of qubits. At this point, there is no long-range entanglement at all. However, after the bulk (blue in \cref{fig:ghz1}) is measured, long-range entanglement is created across the orange blocks (\cref{fig:ghz2}): the subsequent decoding and feedforward operations are only used to correct the state into the desired GHZ state, but cannot create entanglement. Therefore, the GHZ-type entanglement is created by the teleportation effect of measurement. Interestingly, this teleportation effect is robust to noise.

We remark that \cref{thm:ghzinformal} implies \cref{thm:stateprepinformal} as a special case. An interesting future direction is to explore whether there exists other families of interesting states that can be prepared using the ATG framework.

\paragraph{Proof of fault tolerance.} The fault-tolerance of our single-shot logical state preparation scheme has two main components. First, we identify the stabilizer structure of $\ket{\atg}$. As remarked earlier, the fault tolerance arises from redundancies in the ancilla qubits. Concretely, that manifests in the redundancies of the stabilizers of $\ket{\atg}$, which contains certain ``meta-checks''. Second, we adapt the clustering arguments by \cite{Kovalev2012FaultTO, Gottesman2013FaulttolerantQC} for decoding quantum LDPC codes from random errors to our setting, and develop analogous techniques to argue that the meta-checks contain enough redundancy to be robust against random errors.\\


\noindent \textit{The meta-check information}. Following \cite{Bravyi2020Quantum}, inside the $\ket{\atg}$ (which is a stabilizer state), one can define stabilizers which cleanly factorize into a product of $X$ Pauli's on the bulk of the graph state. We refer to these stabilizers as ``meta-checks", since \\

\noindent 1. After measuring the bulk in the $X$ basis, they remain stabilizers of the post-measurement state,\\

\noindent 2. They inbuild redundancies (``syndrome of syndromes") into the measured string $s$.\\

These redundancies will later allow us to use $s$ to infer errors that occur on the bulk, and in turn, correctly decode the error $P(s)$ for the qubits on the boundary. Here we give an example and defer a rigorous exposition to \cref{sec:stabilizers}. Recall that on every even layer $t$ of the $\atg$, there is a copy of the $Z$-Tanner graph lying on the layers $t-1$ and $t+1$ adjacent to it. Now, consider a $Z$-check of the code denoted as $c$, and let $\mathsf{Supp}(c)\subseteq[n]$ be the set of qubits it acts on. We show that the following operator 
\begin{equation}
    X_{(c, t-1)}\otimes X_{(c, t+1)}\bigotimes_{i\in \mathsf{Supp}(c)} X_{(i, t)} \quad (\text{see \cref{fig:meta-check}}),
\end{equation}
which applies a Pauli-$X$ on the copies of the check qubit $c$ in layers $t-1$ and $t+1$, as well as the qubits supported in the check in layer $t$, is a stabilizer of $\ket{\atg}$.





While perhaps a bit abstract at first, this operator captures the change in the $Z$ syndrome after a layer of data-errors on the code and measurement errors on the checks: akin to the redundancy between adjacent syndrome measurements that appear in a repeated syndrome measurement protocol. The fact that these stabilizers factorize cleanly into products of $X$ Paulis is non-trivial and depends carefully on the structure of $\ket{\atg}$, which we discuss further in \cref{sec:stabilizers}.\\

\noindent \textit{The clustering argument.} \cite{Kovalev2012FaultTO, Gottesman2013FaulttolerantQC} showed that any quantum LDPC code can correct random Pauli errors (below a constant threshold noise rate), even in the presence of syndrome measurement errors, by repeating the syndrome measurements $T$ times and running an information-theoretic minimum weight decoder. Their proof reasoned that the data-errors and syndrome-errors could be arranged in a low degree graph known as the \textit{syndrome adjacency graph}, wherein the locations of mismatch errors (where the output of the decoder does not match the true error) clusters into connected components. By a percolation argument on this low-degree graph, one can argue these clusters are unlikely to ``contrive" into configurations that induce logical errors on the code. Roughly speaking, the probability of a logical error is upper bounded by
\begin{equation}
    \sum_{s\geq d} \under{\bigg( \text{\# of Connected Components }|K|=s \bigg)}{\text{A union bound over clusters of size $s$}} \times \under{\quad \bigg(p^{\Omega(s)}\bigg)\quad}{\text{Odds of such a cluster}} \approx \sum_{s\geq d} (p\cdot z)^{\Omega(s)} \approx (pz)^d,
\end{equation}

\noindent where only clusters of size larger than the code distance $d$ can incur logical errors, and $z$ is the degree of the graph.\footnote{A function only of the locality of the LDPC code. We assumed the noise rate is bounded by a threshold $p<z^{-1}$ for the geometric series to converge.}

Here, we broadly follow this strategy and construct an analogous low-degree syndrome adjacency graph which captures where our min-weight decoder incorrectly guesses a data or check error. Unfortunately, making this argument rigorous is easier said than done. Arguably the main challenge lies in relating the size of clusters of mismatch errors, to the number of true physical errors which lie within the cluster, to properly claim the decay rate $p^{\Omega(s)}$ with the cluster size. As we discuss in \cref{section:FT-clustering} and \cref{section:ghz}, this calculation is highly dependent on the boundary conditions of the syndrome adjacency graph, tailored for Bell states and to GHZ states.


\subsection{Discussion and other related works}
\label{section:discussion}

\paragraph{Fault tolerant shallow circuits.} As a direct corollary, our single-shot logical state preparation can be used as the basis of a fault tolerant implementation of a shallow quantum circuit. 
\begin{corollary}
    Fix an arbitrary quantum LDPC code $Q$. Let $C$ be a constant depth quantum circuit which uses $\ket{0}$ or $\ket{+}$ or GHZ initial states, gates which are transversal in $Q$, and measurements in the standard or Hadamard basis. Then there is a constant depth fault tolerant implementation of $C$ which uses a single round of mid-circuit measurement and feedforward.
\end{corollary}

Here we simply run our single-shot logical state preparation scheme to prepare encoded initial states, and then directly apply transversal gates. Since the logical circuit is constant depth, the residual local stochastic noise still remains local stochastic, and the logical $Z$ or $X$ measurements at the end can be decoded using the code.

A natural question is whether we can remove the mid-circuit measurement and feedforward, and implement an ideal shallow quantum circuit purely via a noisy shallow quantum circuit. Ref.~\cite{Bravyi2020Quantum} gives such an example: starting from RBH, they directly apply a logical Clifford circuit \emph{without} correcting the Pauli error $P(s)$. However, since the logical circuit is Clifford, the Pauli error $P(s)$ is propagated through the circuit and corrected at the end. Our result similarly implies a more general version:

\begin{corollary}\label{cor:clifford}
    Fix an arbitrary quantum LDPC code $Q$. Let $C$ be a constant depth quantum circuit which uses $\ket{0}$ or $\ket{+}$ or GHZ initial states, and gates which are Clifford and transversal in $Q$, and measurements in the standard or Hadamard basis. Then there is a constant depth fault tolerant implementation of $C$ using a noisy circuit (without mid-circuit measurements and feedforward).
\end{corollary}

Note that we do not address the issue of decoding efficiency. In particular, although \cref{cor:clifford} does not need mid-circuit decoding, the decoding at the end of the circuit may still be inefficient. An interesting direction is to explore if \cref{cor:clifford} can be the basis of other quantum advantage experiments with noisy shallow quantum circuits against shallow classical circuits, as in~\cite{Bravyi2020Quantum}.

\paragraph{Single-shot preparation of general stabilizer states.} Beyond GHZ state preparation, we show that single-shot logical state preparation can also be achieved for more general stabilizer states, using a combination of the ATG and Steane's logical measurement. Let $\ket{\psi}$ be a $m$-qubit stabilizer state with the following structure: there exists a set of $m$ stabilizer generators that can be divided into two sets $S_X\subseteq\{I,X\}^{\otimes m}$ and $S_Z\subseteq\{I,Z\}^{\otimes m}$, such that either $S_X$ or $S_Z$ has constant weight and degree. The $m$-qubit GHZ state is such an example, where $S_X=\{X^{\otimes m}\}$ and $S_Z=\{Z_1 Z_2, Z_2 Z_3,\dots,Z_{m-1}Z_m\}$. We prove fault tolerant single-shot preparation for any such state $\ket{\psi}$ encoded across $m$ copies of an arbitrary quantum LDPC code. See \cref{sec:cssstateprep} for more details.

\paragraph{Single-shot quantum error correction.} A closely related concept to our work is single-shot quantum error correction~\cite{Bombin2014SingleShotFQ}: consider an encoded code block subject to random errors, it is shown that for certain codes a single round of syndrome measurement suffices to correct the error (instead of performing repeated syndrome measurements), even in the presence of syndrome measurement errors. Examples include high-dimensional color codes~\cite{Bombin2013GaugeCC, Bombin2014SingleShotFQ} and toric codes~\cite{Kubica2021SingleshotQE}, hypergraph product codes \cite{Fawzi2017EfficientDO}, quantum Tanner codes \cite{Gu2023SingleshotDO} and hyperbolic codes \cite{Breuckmann2020SingleShotDO}.

As the ATG is closely related to repeated syndrome measurements, it is natural to ask whether single-shot logical state preparation with \emph{constant} space overhead exists for those codes, or more generally, for any code with single-shot error correction. While this is proven for 3D gauge color codes~\cite{Bombin2013GaugeCC, Bombin2014SingleShotFQ}, it is unclear whether a general reduction from single-shot logical state preparation to single-shot error correction exists. The key issue is that single-shot error correction is defined in the context of correcting errors on an existing encoded code block, while single-shot logical state preparation is about preparing an encoded code state \emph{from scratch}. We leave further exploration of this issue for future work; it is an interesting question in the context of reducing the spacetime overhead of quantum fault tolerance (also see~\cite{zhou2024algorithmic}). 

\subsection{Organization}

We organize the rest of this work as follows. We begin with basic definitions of quantum error correction in \cref{sec:prelim}. In \cref{section:atg}, we discuss the stabilizer structure of the alternating Tanner graph state. In \cref{section:sssp}, we present our single-shot logical state preparation algorithm. Finally, in \cref{section:FT-clustering}, we present our proof of fault-tolerance via clustering (\cref{thm:stateprepinformal}). 

We defer to \cref{section:factorize} omitted proofs from \cref{section:atg}, \cref{section:mbqc} our reduction to repeated measurements based on MBQC (\cref{thm:stateprepmbqcinformal}), and in \cref{section:ghz} we present single shot state preparation for the encoded $\ghz$ state (\cref{thm:ghzinformal}).





\section{Preliminaries}
\label{sec:prelim}

The set of Pauli operators on a set of $n$ qubits is denoted as $\mathsf{Pauli}(n)$. Our fault-tolerance theorems are formalized in the ``local stochastic" model of Pauli noise:

\begin{definition}
    A Pauli error $E\in \mathsf{Pauli}(n)$ is said to be a local stochastic error of noise rate $p\in [0, 1]$, or, $E\leftarrow N(p)$ if 
    \begin{equation}
    \forall S\subset [n],\quad \mathbb{P}_{E\leftarrow N(p)}\big[S\subset \mathsf{supp}(E)\big] \leq p^{|S|}.
\end{equation}
\end{definition}

\noindent We present logical state preparation algorithms for a well-known class of quantum error correcting codes known as CSS codes \cite{Calderbank1996GoodQE, Steane1996SimpleQE}.

\begin{definition}
    Let $H^X\in \mathbb{F}_2^{m_x\times n}, H^Z\in \mathbb{F}_2^{m_z\times n}$ be two full-rank parity check matrices satisfying $H^X (H^Z)^\dagger = 0$. The $[[n, k, d]]$ CSS code $(H^X, H^Z)$ is the joint $+1$ eigenspace of the set of commuting Pauli operators $X^a, Z^b$, where $a, b\in \mathbb{F}_2^n$ are rows of $H^X, H^Z$ respectively.  
\end{definition}

The number of logical qubits $k = n-m_x-m_z$ encoded into the CSS code is determined by the number of linearly independent $X$ and $Z$ checks. Let us denote the linear subspace $C_Z = \{y\in\mathbb{F}_2^n: H^X\cdot y = 0\}$ (resp, $C_X$), and its dual $C_Z^\perp$ be the code spanned by the rows of $H^X$. The distance of the CSS code is then 
\begin{equation}
    d = \min_{y\in (C_X\setminus C_Z^\perp)\cup (C_Z\setminus C_X^\perp)} |y|.
\end{equation}

\noindent We say $i\sim c$ if $i\in [n]$ is in the support of the check $c\in [m]$, i.e. $H_{c,i}=1$.

\begin{definition}
    [Tanner Graph]\label{def:tanner} The tanner graph of a parity check matrix $H\in \mathbb{F}_2^{m\times n}$ is the bipartite graph on the vertex set $[m]\cup [n]$, with edge set defined by the relation $i\sim c$. 
\end{definition}

We refer to the tanner graphs of $H^X, H^Z$ as $G^X, G^Z$. Our results make reference to a particular class of CSS codes known as LDPC codes, where the tanner graphs are sparse.

\begin{definition}
    A CSS code $Q$ is said to be an $\ell$-LDPC code if there exists a choice of checks $\in \mathsf{Pauli}(n)$ for $Q$ with Pauli weight $\leq \ell$ and such that the number of check acting non-trivially on any fixed qubit is $\leq \ell$. 
\end{definition}

\section{The Alternating Tanner Graph State}
\label{section:atg}

We dedicate this section to a definition of the alternating tanner graph state $\ket{\atg}$, and a study of its properties. Starting from a generic quantum LDPC code $Q$, we will define the cluster state $\ket{\atg}$ simply by specifying an undirected graph $G = (V, E)$. So long as the degree of $G$ is bounded, one can find a circuit $W$ which prepares the cluster state in low-depth: 

\begin{equation}\label{equation:graph_state}
    \ket{\atg} = \bigg(\prod_{(u, v)\in E} \mathsf{CZ}_{u, v} \bigg) \ket{+}^{\otimes |V|}
\end{equation}

\subsection{The Alternating Tanner Graph}

Fix an arbitrary integer $T \geq 1.$ $G$ will be defined on $(2T+1)$ layers of copies of the LDPC code, which alternate X and Z checks. That is to say, at each layer, we will place a copy of the tanner graph of $H^X$ or $H^Z$ (\cref{def:tanner}).

\begin{tcolorbox}

\textbf{The Alternating Tanner Graph.} Of the $2T+1$ layers, $\{1, 2, \cdots, 2T+1\}$

\begin{itemize}
    \item[--] Every layer has a copy of the $n$ ``code" qubits. Each code qubit $i\in [n]$ at layer $t$ is connected to its copy above and below it $(i, t\pm 1)$.

    \item[--] \textit{Z Layers.} Odd layers have a copy of $m_z$ ``Z" parity check qubits. A copy of the $Z$ tanner graph is placed between the $n$ code qubits and the $m_z$ Z parity check qubits. 

    \item[--] \textit{X Layers.} Even layers have a copy of $m_x$ ``X" parity check qubits. A copy of the $X$ tanner graph is placed between the $n$ code qubits and the $m_x$ X parity check qubits. 
\end{itemize}

\end{tcolorbox}

One should picture the copies of the code qubits stacked vertically in 1D layers, while the X and Z ancilla qubits lie to their left and right respectively. We refer the reader back to \cref{fig:atg} in the introduction. Henceforth we will index the vertices in $G$ as tuples, $(i, t)$ for any code qubit $i\in [n]$ at layer $t$, and $(c, t)$ for check qubits $c\in [m_x] $ (or $m_z$). We highlight that the only vertical connections in the edge set $E$ are between copies of the same code qubit. 

It will become relevant to partition the vertex/qubit set $V$ into Bulk and boundary qubits $V = \mathcal{B}\cup \partial$, as alluded to in the overview. The boundary qubits refer to the $2n$ code qubits lying on the 1st and $(2T+1)$th layers. 

\subsection{The Stabilizers of the \texorpdfstring{$\ket{\atg}$}{|ATG>}}
\label{sec:stabilizers}

The alternating Tanner graph state $\ket{\atg}$ is a stabilizer state. A complete set of stabilizer generators is defined as follows: for each $u\in V$, there is an associated stabilizer
\begin{equation}
    G_u  = X_u \bigotimes_{v:\,(u, v)\in E} Z_v.
\end{equation}

Suppose we would like to prepare the encoded Bell state $\ket{\overline{\Phi}}=\frac{1}{\sqrt{2}}\left(\ket{\overline{00}}+\ket{\overline{11}}\right)$ across two copies of the code. We will use this collection of graph-state stabilizers to define stabilizers of the post-measurement state. Following \cite{Bravyi2020Quantum}, we will identify two subgroups of the stabilizer group $\mathcal{S}$ of $\ket{\atg}$, $\mathcal{S}^0 \subset \mathcal{S}^1\subset \mathcal{S}$, satisfying the following constraints:

\begin{enumerate}[(i)]
    \item Any element of $\mathcal{S}^0$ can be written as $\mathbb{I}_{\partial}\otimes X(\alpha)_{ \mathcal{B}}$ on some subset $\alpha\in \{0, 1\}^ \mathcal{B}$.

    \item Any element of $\mathcal{S}^1$ can be written as $S_{\partial}\otimes X(\alpha)_{ \mathcal{B}}$ on some subset $\alpha\in \{0, 1\}^ \mathcal{B}$, and for some stabilizer $S$ of $\ket{\Bar{\Phi}}$.

    \item For every stabilizer $S$ of $\ket{\Bar{\Phi}}$, there exists an element of $\mathcal{S}^1$ of the form $S_{\partial}\otimes X(\alpha)_{ \mathcal{B}}$.
\end{enumerate}

The fact that these subgroups act only as Pauli X's on the Bulk implies that after an X basis measurement, they remain stabilizers of the post-measurement state, and we can recover their information from the measured string $s$. In the next section, we will describe how to use the redundancies in these stabilizers to infer the relevant Pauli frame correction. Here, in this section, we simply show how to define these stabilizers using the structure of the tanner graph state. We defer to \cref{section:factorize} rigorous proofs that these stabilizers factor as above. Let us begin with $\mathcal{S}^0$. \\

\begin{figure}[t]
\begin{subfigure}[b]{0.33\textwidth}
\centering
     \includegraphics[width=1.0\linewidth]{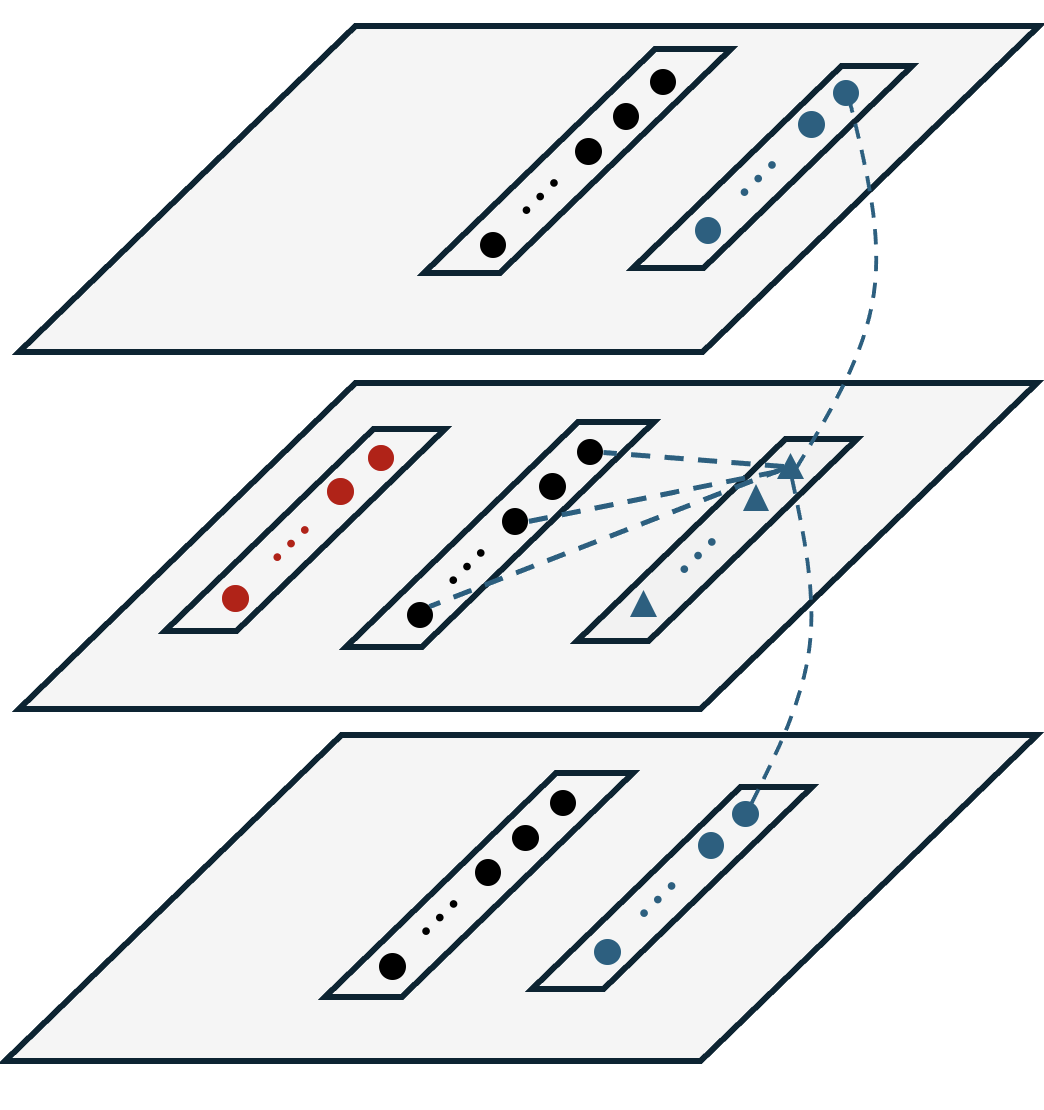}
     \caption{A Z meta-check, on an \\ even layer (blue triangles).}
    \label{fig:meta-check}
\end{subfigure}
\begin{subfigure}{0.33\textwidth}
\centering
    \includegraphics[width=0.95\linewidth]{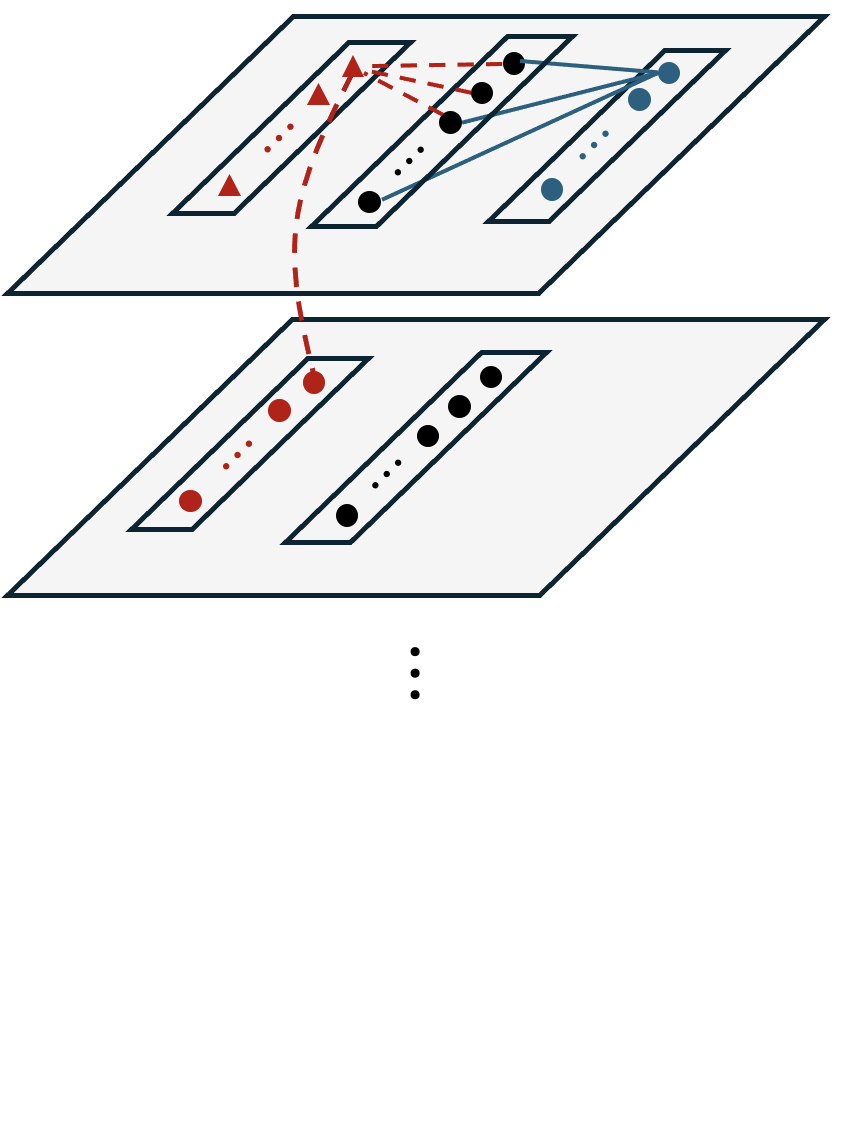}
    \caption{The stabilizers of the \\top boundary code. }
    \label{fig:stab-boundary}
\end{subfigure}
\begin{subfigure}{0.33\textwidth}
\centering
    \includegraphics[width=1.0\linewidth]{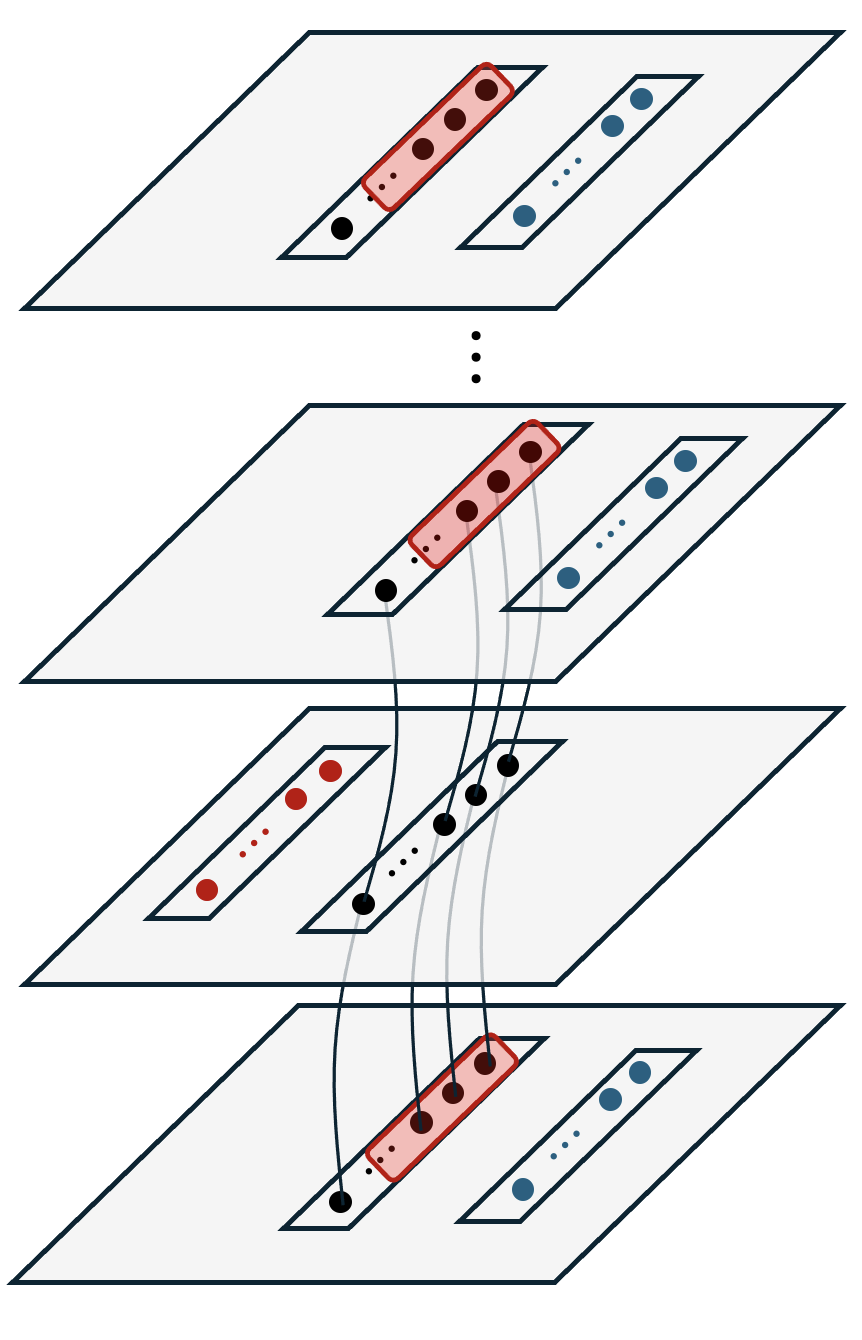}
    \caption{An Encoded $\bar{X}\otimes \bar{X}$ \\ Logical Stabilizer.}
    \label{fig:stab-encoded}
\end{subfigure}
    \caption{The Stabilizers of the Graph State $\ket{\atg}$}
\end{figure}

\noindent \textbf{The Meta-Checks $\mathcal{S}^0$.} $\mathcal{S}^0$ will consist of X or Z ``meta-checks" which encode redundancies into the bulk qubits. Each meta-check will be centered around an ``meta-vertex" (the triangles in \cref{fig:meta-check}) - the support of such meta-checks will consist of copies of the X and Z tanner graphs in alternating layers (offset from those of $E$); together with vertical connections between copies of the same ancillas. For each even layer $t$ and $c\in [m_z]$, we place a $Z$ meta-check:

\begin{equation}
    G_{(c, t-1)}\cdot G_{(c, t+1)}\cdot \prod_{i\sim c} G_{(i, t)} = X_{(c, t-1)}\otimes X_{(c, t+1)}\bigotimes_{i\sim c} X_{(i, t)} \quad (\text{\cref{fig:meta-check}, blue})
\end{equation}

\noindent In turn, for each odd layer $t$ and $c\in [m_x]$, we place a $X$ meta-check:

\begin{equation}
    G_{(c, t-1)}\cdot G_{(c, t+1)}\cdot \prod_{i\sim c} G_{(i, t)} = X_{(c, t-1)}\otimes X_{(c, t+1)}\bigotimes_{i\sim c} X_{(i, t)} 
\end{equation}

The fact that these products of graph state factorize cleanly into products of X operators is non-trivial, and carefully leverages the fact that $Q$ is a CSS code. This observation can first be traced back to \cite{Bolt2016Foliated}. While we defer a rigorous proof to \cref{section:factorize}, the claim is that each X ancilla qubit $d\in [m_x]$ which arises in the neighborhood of the support of any given Z meta-check $c\in [m_z]$, appears precisely an even number of times. This is since $d$ and $c$ on layer $t$ are connected through a qubit $i$ iff $H^Z_{c, i}\cdot H^X_{d, i} = 1$, and therefore the number of appearances is

\begin{equation}
        \sum_i H^Z_{c, i}\cdot H^X_{d, i} = \big(H^Z (H^X)^T\big)_{c, d} = 0 \mod 2,
    \end{equation}

\noindent which is $0$ since $H^X, H^Z$ define a CSS code.

The stabilizers in $\mathcal{S}^1$ arise in two types. Recall that $\ket{\Bar{\Phi}}$ consists of an encoded maximally entangled state across two copies of the LDPC code $Q$. Then, within $\mathcal{S}^1$ there will be stabilizers of the individual boundary codes, and encoded stabilizers of the Bell state.\\

\noindent \textbf{The Stabilizers of the Boundary Codes.} We specify the X and Z stabilizers of the boundary codes as follows. To begin, let us consider the simpler case, consisting of the $Z$-type stabilizers. For each $c\in [m_z]$, there exists graph state stabilizers $\in \mathcal{S}$, satisfying the following decomposition

    \begin{equation}
        G_{(c, 1)} = X_{(c, 1)} \bigotimes_{i\sim c} Z_{(i, 1)}, \quad G_{(c, 2T+1)} = X_{(c, 2T+1)} \bigotimes_{i\sim c} Z_{(i, 2T+1)} \quad (\text{\cref{fig:stab-boundary}, blue})
    \end{equation}

    \noindent Which directly follows from the definition of the graph state stabilizers. Note that these operators act as $Z$ Paulis on the boundary $\partial$ and $X$ on $\mathcal{B}$.
    
    The X-type stabilizers on the boundary codes are more complicated, and require products of graph state stabilizers. For each $X$-type stabilizer $c\in [m_x]$ of $Q$, there exists products of graph state stabilizers $\in \mathcal{S}$, satisfying the decomposition
    \begin{gather}
        G_{(c, 2)} \cdot \prod_{i\sim c} G_{(i, 1)} =  X_{(c, 2)} \bigotimes_{i\sim c} X_{(i, 1)}, \\ \quad G_{(c, 2T)} \cdot \prod_{i\sim c} G_{(i, 2T+1)} =  X_{(c, 2T)} \bigotimes_{i\sim c} X_{(i, 2T+1)}\quad (\text{\cref{fig:stab-boundary}, red}).
    \end{gather}

    \noindent Which, we note, act as $X$ stabilizers on the boundary codes $\partial$ in tensor product with an X Pauli on the Bulk, as desired.\\

    \noindent \textbf{The Encoded Logical Stabilizers.} The encoded Bell pairs are stabilized by products $\Bar{X}_1\otimes \Bar{X}_{2T+1}, \Bar{Z}_1\otimes \Bar{Z}_{2T+1}$ of logical operators. We construct these operators using graph state stabilizers in $G$, and only $X$ operators on the Bulk, via products of stabilizers in alternating layers of $G$.

    To begin, let us consider the encoded $\Bar{X}_1\otimes \Bar{X}_{2T+1}$ stabilizer. Let $\alpha_x\subset [n]$ denote the support of a logical $\Bar{X}$ on $Q$. Then, we can write the  $\Bar{X}_1\otimes \Bar{X}_{2T+1}$ stabilizer as:
    
    \begin{equation}
        \prod_{\substack{i\in \alpha_x \\ t \text{ odd }}} G_{(i, t)} =  \bigotimes_{\substack{i\in \alpha_x \\ t \text{ odd }}}X_{(i, t)}= \Bar{X}_1\otimes \Bar{X}_{2T+1}  \bigotimes_{\substack{i\in \alpha_x \\ t \in \{3, 5,\cdots\}}}X_{(i, t)}
    \end{equation}

    \noindent Similarly, let $\alpha_z\subset [n]$ denote the support of a logical $\Bar{Z}$ on $Q$. Then, we can write the $\Bar{Z}_1\otimes \Bar{Z}_{2T+1}$ stabilizer as:

    \begin{equation}
        \prod_{\substack{i\in \alpha_z \\ t \text{ even }}} G_{(i, t)} = \Bar{Z}_1\otimes \Bar{Z}_{2T+1} \bigotimes_{\substack{i\in \alpha_z \\ t \text{ even }}} X_{(i, t)} \quad (\text{\cref{fig:stab-encoded}})
    \end{equation}

    To argue that these operators factor as desired, we similarly apply the constraint that $Q$ defines a CSS code. We refer the reader to \cref{section:factorize} for the proofs.

\section{Single-shot Logical State Preparation}
\label{section:sssp}

In this section, we overview our fault-tolerant single-shot state preparation algorithm for any LDPC CSS code. Our algorithm is based on that of \cite{Bravyi2020Quantum} and is comprised of three general steps, which we summarize in \cref{subsection:algorithm}. Subsequently, we rigorously define what it means for this process to be fault tolerant, in terms of abstract ``recovery" and ``repair" functions $\rec, \rep$ which capture the Pauli frame computation and the residual stochastic error on the code. In \cref{section:rec/rep}, we conclude this section by instantiating $\rec, \rep$. 

In the ensuing section \cref{section:FT-clustering}, we prove a threshold for our single-shot state preparation decoder.

\subsection{The Algorithm}
\label{subsection:algorithm}

As mentioned, our single-shot state preparation algorithm is based on that of \cite{Bravyi2020Quantum} and is comprised of three general steps, summarized below.

\begin{tcolorbox}

\textbf{Single-Shot State Preparation.}

\begin{enumerate}
    \item Using a constant-depth circuit $W$, we prepare a graph state $\ket{\atg} = W\ket{0^V}$. $G = (V, E)$ is defined on ``bulk'' qubits $\mathcal{B}$ and ``boundary'' qubits $\partial = V\setminus \mathcal{B}$. 

    \item Measure all the bulk qubits $\mathcal{B}$ in the Hadamard basis $\ket{\pm}$, resulting in a string $s$.

    \item Using $s$, compute a Pauli frame correction $\rec(s)\in \mathsf{Pauli}(\partial)$, and adaptively apply it to the boundary qubits $\partial$. The resulting (unnormalized) state is given by:
    \begin{equation}
        \bigg(\ketbra{\pm_s}_{\mathcal{B}}\otimes \rec(s)_\partial\bigg) \ket{\atg}
    \end{equation}
\end{enumerate}

\end{tcolorbox}

In the absence of errors, we design $\rec(s)$ to ensure that the resulting state is an encoded state of interest (such as $\ket{\bar{0}}, \ket{\bar{+}}, \ket{\overline{\ghz}}, \cdots$). In the next subsections, we will focus on preparing $\ket{\overline{\Phi^{\otimes k}}}$, $k$ copies of EPR pairs encoded across two copies of the LDPC code -- ``the boundaries".\footnote{By performing a logical measurement on one of the surfaces of the encoded Bell state in the standard or Hadamard basis, and performing a Pauli correction, one can prepare any encoded state $\ket{\bar{x}}$ where $x_i\in \{0, 1, +, -\}$.} To ensure that this state-preparation algorithm is fault-tolerant, we stipulate that even in the presence of random errors during the execution of $W$ and the measurements, the resulting output state is statistically close to $\ket{\overline{\Phi^{\otimes k}}}$ up to local stochastic noise. 

\begin{definition}
    [cf.~\cite{Bravyi2020Quantum}] \label{def:sssp} A family of $[[n, k, d]]$ stabilizer codes $Q$ admits a \emph{Single-Shot State Preparation} procedure for an encoded state $\ket{\bar{\psi}}$ if there exists a constant depth circuit $W$ on a set of qubits $\partial\cup \mathcal{B}$, and deterministic recovery and repair functions, 
    \begin{align}
        &\rec: \{0, 1\}^{\mathcal{B}}\rightarrow \mathsf{Pauli}(\partial) \\
         &\rep: \mathsf{Pauli}(\mathcal{B})\rightarrow \mathsf{Pauli}(\partial), 
    \end{align}
    \noindent such that, for any error on the Bulk $P\in \mathsf{Pauli}(\mathcal{B})$ and measurement outcome $s \in \{0, 1\}^{\mathcal{B}}$,
        \begin{align}
     \bigg(\ketbra{\pm_s}\otimes \rec(s)\bigg) P W\ket{0}^{\mathcal{B}}\otimes \ket{0}^{\partial}  = \gamma_s \ket{\pm_s} \otimes \rep(P) \ket{\bar{\psi}},
    \end{align}
    where $\gamma_s\in\mathbb{C}$. Moreover, if $P\leftarrow \mathcal{N}(p)$ is a local stochastic error on the Bulk, then $\rep(P)\leftarrow \mathcal{N}(c_1\cdot p^{c_2})$ is a local stochastic Pauli error on $\partial$, for two constants $c_1, c_2$.
\end{definition}

Naturally, in the presence of measurement errors, one cannot hope to perfectly prepare the encoded logical state. Instead, the function $\rep$ quantifies the \textit{residual error} on the ideal state. The state preparation procedure is then fault-tolerant if it manages to convert local stochastic noise during the state preparation circuit into local stochastic noise on the output state.

\subsection{The Recover and Repair functions}
\label{section:rec/rep}

We are now in a position to define the Pauli frame correction $\rec$, and the residual Pauli noise $\rep$.

\subsubsection{The Pauli Frame Correction \texorpdfstring{$\rec$}{Rec}}

The definition of $\rec$ proceeds in two steps. At a high level, we begin by leveraging the meta-check information $\mathcal{S}^0$, to make a guess $Z(\beta)$ for the error which occurs on the Bulk $\mathcal{B}$.\footnote{Note that we can restrict ourselves to $Z$ errors on $\mathcal{B}$, since these qubits are measured in the X basis.} Then, we pick $\rec$ to be an arbitrary Pauli error supported on the boundary $\partial$, which is consistent with the information from $\mathcal{S}^1$, and the inferred error $Z(\beta)$.

To understand the role of the meta-checks in this sketch, let us concretely show how the measurement outcome string $s\in \{0, 1\}^{\mathcal{B}}$ allows us to extract partial information about any Z-type Pauli error on the Bulk $P = Z(\eta)_\mathcal{B}$. For any stabilizer $\mathbb{I}_\partial\otimes X(\alpha)_\mathcal{B}\in \mathcal{S}^0$, we have 
\begin{equation}
\begin{aligned}
&(-1)^{s\cdot \alpha}  \bigg(\ketbra{\pm_s}\otimes \mathbb{I}_\partial\bigg) Z(\eta)_{ \mathcal{B}} \otimes \mathbb{I}_{\partial} \ket{\atg}\\ &=  \bigg(\ketbra{\pm_s}\otimes \mathbb{I}_\partial\bigg) X(\alpha)Z(\eta)_{ \mathcal{B}} \otimes \mathbb{I}_{\partial} \ket{\atg}\\ 
&=  (-1)^{\alpha\cdot \eta} \bigg(\ketbra{\pm_s}\otimes \mathbb{I}_\partial\bigg) Z(\eta)_{ \mathcal{B}} \otimes \mathbb{I}_{\partial} \ket{\atg}, 
\end{aligned}
\end{equation}

\noindent therefore revealing the ``syndrome" information $\eta\cdot \alpha = s\cdot \alpha\mod 2.$ By collecting this syndrome information, we can make a guess for the error on the Bulk, as described in Step 2 of \cref{fig:rec}. An analogous calculation can be reproduced for the stabilizers $S_\partial\otimes X(\alpha)_\mathcal{B}\in \mathcal{S}^1$:

\begin{equation}
\begin{aligned}
& S_\partial\cdot \bigg(\ketbra{\pm_s}\otimes \mathbb{I}_{\partial}\bigg)  Z(\eta)_{ \mathcal{B}} \otimes \mathbb{I}_{\partial} \ket{\atg}\\ 
&=  (-1)^{s\cdot \alpha} \bigg(\ketbra{\pm_s}\otimes \mathbb{I}_\partial\bigg) \bigg(X(\alpha)Z(\eta)_{ \mathcal{B}} \otimes S_{\partial}\bigg) \ket{\atg}\\ 
&=  (-1)^{\alpha\cdot \eta+\alpha\cdot s} \bigg(\ketbra{\pm_s}\otimes \mathbb{I}_\partial\bigg) Z(\eta)_{ \mathcal{B}} \otimes \mathbb{I}_{\partial} \ket{\atg}. 
\end{aligned}
\end{equation}

In this manner, if we were to perform an ideal syndrome measurement of $S_\partial$ on the post-measurement state, the value readoff would be $\alpha\cdot \eta+\alpha\cdot s$. The goal is to correct all of them to 0 (or $+1$ for the stabilizer measurement).

\begin{figure}[t]
\centering
\begin{tcolorbox}

\textbf{Recover $\rec(s)$.} Given a measurement string $s\in \{0, 1\}^{\mathcal{B}},$

\begin{enumerate}
    \item For each stabilizer $\mathbb{I}_{\partial}\otimes X(\alpha)_{ \mathcal{B}} \in \mathcal{S}^0$, compute its syndrome $s_\alpha = s\cdot \alpha$.

    \item Find the minimum-weight Z-type Pauli $Z(\beta)$ supported on $\mathcal{B}$, consistent with the syndromes $s_\alpha$ of $\mathcal{S}^0$. 

    \item For each stabilizer $\mathcal{S}_{\partial}\otimes X(\gamma)_{ \mathcal{B}} \in \mathcal{S}^1$, compute the \textit{corrected} syndrome $$s_\gamma' = s\cdot \gamma\oplus \gamma\cdot \beta.$$

    \item Let $\rec(s)$ be an \emph{arbitrary} Pauli supported on $\partial$, consistent with the computed corrected syndromes $s'$ of $\mathcal{S}^1$. \footnote{Here, for simplicity we assume the CSS code is specified by full rank matrices $H_X, H_Z$, such that there always exists an error associated to every syndrome vector.}

\end{enumerate}
\end{tcolorbox}
\caption{The Pauli Frame}
\label{fig:rec}
\end{figure}

Unfortunately, our decoder does not have access to these ideal values, and instead can only make a guess of them using $s$, and the inferred error $Z(\beta)$, as done in step 3. Note that this step is not necessarily efficient.

After $\rec(s)$ is applied, the state equals $\bigg(\ketbra{\pm_s}\otimes \rec(s)_{\partial}\bigg)  Z(\eta)_{ \mathcal{B}} \otimes \mathbb{I}_{\partial} \ket{\atg}$ up to normalization. For each code stabilizer $S_\partial$, suppose the corresponding cluster state stabilizer is $S_\partial\otimes X(\alpha)_\mathcal{B}\in \mathcal{S}^1$, then the \emph{residual syndrome} is given by
\begin{equation}
    \begin{aligned}
        & S_\partial\cdot \bigg(\ketbra{\pm_s}\otimes \rec(s)_{\partial}\bigg)  Z(\eta)_{ \mathcal{B}} \otimes \mathbb{I}_{\partial} \ket{\atg}\\
        &= (-1)^{\alpha\cdot \beta + \alpha\cdot s}\rec(s)_{\partial}\cdot S_{\partial}\cdot \bigg(\ketbra{\pm_s}\otimes \mathbb{I}_\partial \bigg)  Z(\eta)_{ \mathcal{B}} \otimes \mathbb{I}_{\partial} \ket{\atg}\\
        &=(-1)^{\alpha\cdot \beta + \alpha\cdot s}\rec(s)_{\partial}\cdot (-1)^{\alpha\cdot \eta+\alpha\cdot s} \bigg(\ketbra{\pm_s}\otimes \mathbb{I}_\partial\bigg) Z(\eta)_{ \mathcal{B}} \otimes \mathbb{I}_{\partial} \ket{\atg}\\
        &=(-1)^{\alpha\cdot \beta + \alpha\cdot \eta}\bigg(\ketbra{\pm_s}\otimes \rec(s)_{\partial}\bigg)  Z(\eta)_{ \mathcal{B}} \otimes \mathbb{I}_{\partial} \ket{\atg},
    \end{aligned}
\end{equation}
that is, the residual syndrome is given by $\alpha\cdot \beta + \alpha\cdot \eta$. This constitutes a proof that the decoding process in $\rec(s)$ (step 4 of \cref{fig:rec}) can be arbitrary, because the residual syndrome only depends on $\beta$.

\subsubsection{The Residual Error \texorpdfstring{$\rep$}{Rep}}

Once the $\rec(s)$ has been computed and applied, we pick $\rep$ to be a carefully designed Pauli operator (residual error) on $\partial$ satisfying:
\begin{equation}
    \bigg(\ketbra{\pm_s}\otimes \rec(s)\bigg) P \ket{\atg}  = \gamma_s \ket{s} \otimes \rep(P) \ket{\Bar{\Phi}}
\end{equation}

Note that the decoder need not know what $\rep$ is; however, should they be able to perform a \textit{perfect} syndrome measurement after $\rec$ is applied, then $\rep$ can essentially be understood as the minimum weight operator consistent with that syndrome. To be more concrete, we define $\rep(P)$ as a product of 4 terms:

\begin{equation}
    \rep(P) = \rep_X(P)\cdot \rep_Z(P) \cdot \rep_{\bar{X}}(P)\cdot \rep_{\bar{Z}}(P)
\end{equation}

Each of the terms above will correspond to the residual correction of the syndrome of a given stabilizer of $\bar{\Phi}$. If we partition the stabilizers $S_\partial\otimes X(\alpha)\in \mathcal{S}^1$, based on whether $S_\partial$ is an $X$ or $Z$ stabilizer of the LDPC code $Q$, or whether it is a logical $XX$ or $ZZ$ stabilizer of the encoded Bell state, then 

\begin{itemize}
    \item $\rep_X(P)$ is the minimum weight $Z$-type Pauli operator on $\partial$ consistent with the residual syndrome of the $X$ stabilizers of $Q$;
    \item $\rep_Z(P)$ is the minimum weight $X$-type Pauli operator on $\partial$ consistent with the residual syndrome of the $Z$ stabilizers of $Q$;
    \item $\rep_{\bar{X}}(P)\in \{\mathbb{I}, \bar{Z}\}$ is the residual $Z$ logical error, which ensures $\rep(P)$ and $Z(\beta) P$ have the same syndrome under the encoded $XX$ stabilizer. 
    \item $\rep_{\bar{Z}}(P)\in \{\mathbb{I}, \bar{X}\}$ is the residual $X$ logical error, which ensures $\rep(P)$ and $Z(\beta) P$ have the same syndrome under the encoded $ZZ$ stabilizer. 
\end{itemize}

We will argue in the following section that so long as $P$ is a local stochastic error, then $\rep_X(P)$,  $\rep_Z(P)$ are local stochastic as well. Moreover, $\rep_{\bar{X}}(P)$, $\rep_{\bar{Z}}(P)$ are trivial (identity) with high probability.

\section{Proof of fault-tolerance via clustering}
\label{section:FT-clustering}

The main result of this work is the following theorem:

\begin{theorem}\label{theorem:main}
    Fix an integer $T\geq 1$. Let $Q$ be any $[[n, k, d]]$ CSS code, which is $\ell$-LDPC. Then, $Q$ admits a single-shot state preparation procedure for the encoded Bell state $\ket{\Bar{\Phi}}$ using a circuit $W$ on $O(\ell\cdot n\cdot T)$ qubits and of depth $O(\ell^2)$, satisfying the following guarantees:

    \begin{enumerate}
        \item There exists a constant $p^*(\ell) \in (0, 1)$, s.t. if $W$ is subject to local stochastic noise $\mathcal{N}(p)$ of rate $p<p^*$, the state preparation procedure succeeds with probability at least  $1 - n\cdot T\cdot (\frac{p}{p^*})^{\Omega(\min(T, d))}$.
        \item Conditioned on this event, the resulting state is subject to local stochastic noise of rate $O(p^{1/2})$.
    \end{enumerate}
\end{theorem}

This section is the main technical part of this paper, in which we show that if $P$ is a local stochastic error, then the residual error $\rep(P)$ is also local stochastic. Our proof approach follows closely that of \cite{Kovalev2012FaultTO} and \cite{Gottesman2013FaulttolerantQC}, who argued that quantum LDPC codes can (information-theoretically) correct from random errors of rate less than a critical threshold $p< p^*$, even in the presence of syndrome measurement errors. Their key idea was a ``clustering" argument, which reasons that stochastic errors on LDPC codes cluster into connected components on a certain low-degree graph. Using a percolation argument on this low-degree graph, they prove these errors are unlikely to accumulate into a logical error on the codespace.

\subsection{The syndrome adjacency graphs.}

\begin{figure}[btp]

\begin{subfigure}[b]{0.5\textwidth}
\centering
    \includegraphics[width = 0.6\linewidth]{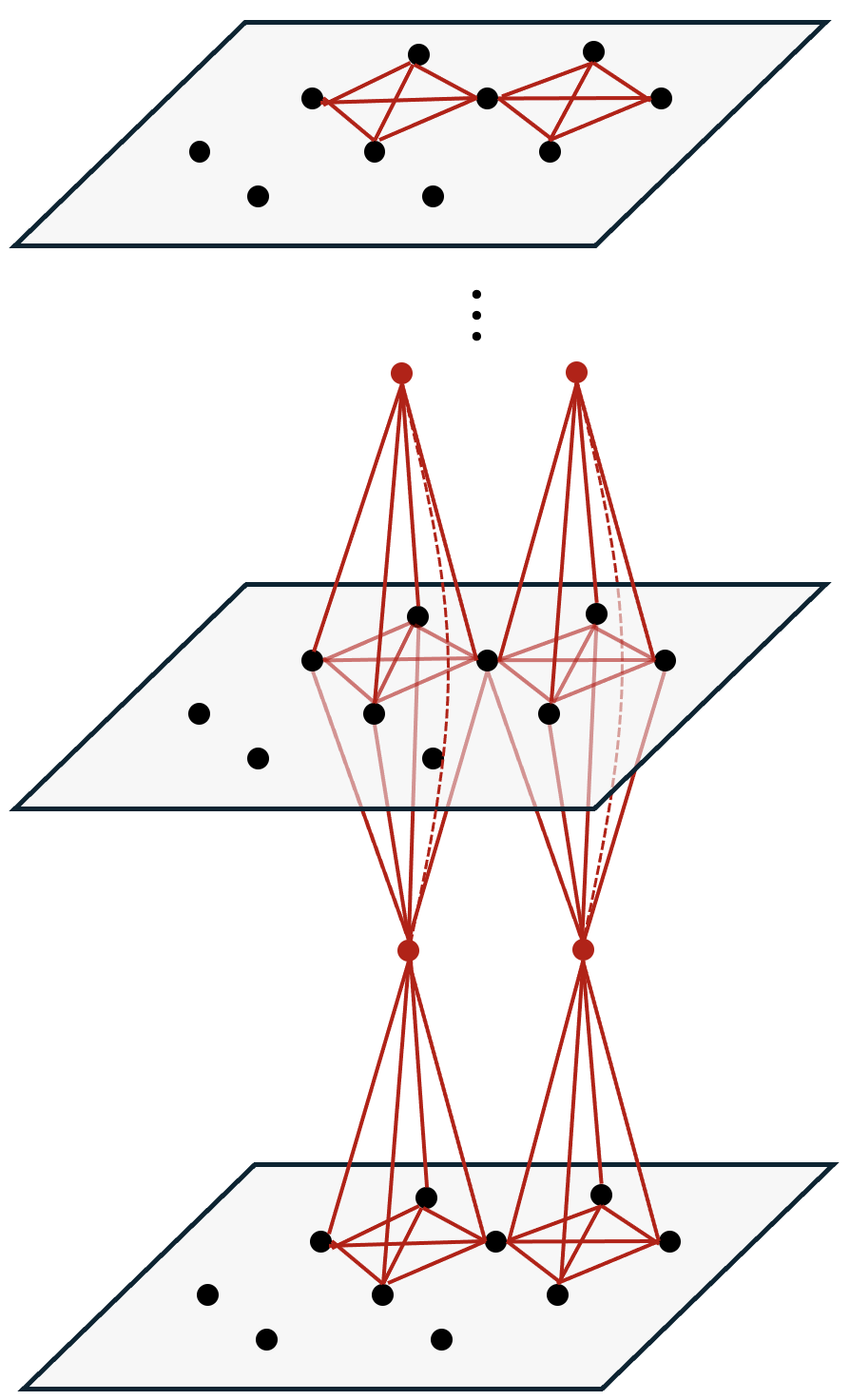}
    \caption{The X Syndrome Adjacency Graph}
    \label{fig:xsag}
\end{subfigure}
\begin{subfigure}[b]{0.5\textwidth}
\centering
    \includegraphics[width = 0.6\linewidth]{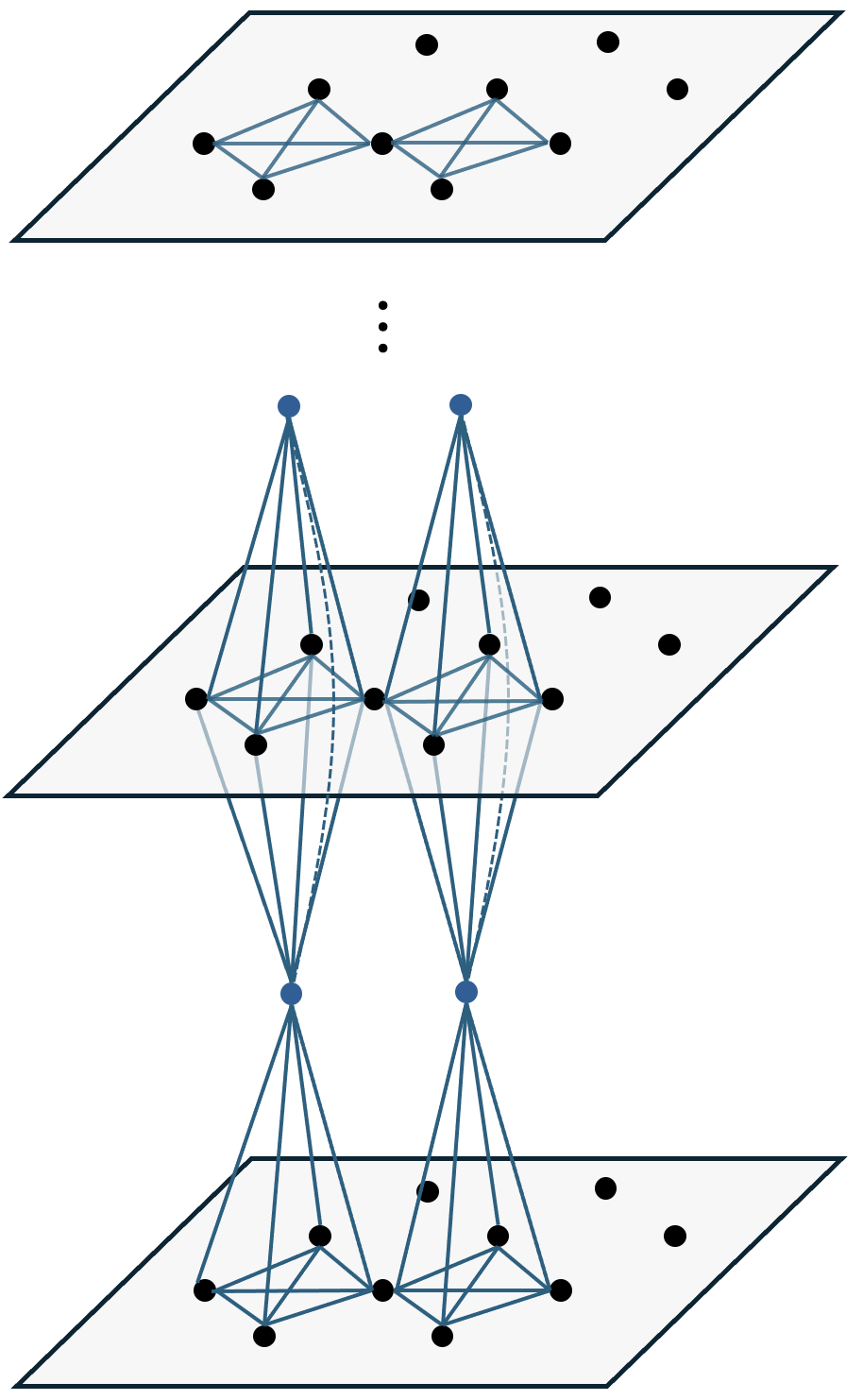}
    \caption{The Z Syndrome Adjacency Graph}
    \label{fig:zsag}
\end{subfigure}
\caption{The nodes of the syndrome adjacency graphs correspond to nodes of the alternating tanner graph. The dashed lines correspond to the ``skip" connections between vertices associated to copies of the same parity check. }
\label{fig:sag}
\end{figure}

 To setup notation, let $P_2, P_3, \cdots P_{2T} \in \{0, 1\}^n$ denote the support of the physical $Z$ errors that occur on each layer of the Bulk qubits, and $B_1, B_3, \cdots B_{2T+1} \in \{0, 1\}^{m_z}, B_2, B_4, \cdots B_{2T} \in \{0, 1\}^{m_x}$ denote the support of the $Z$ errors that occur the ancilla qubits in $G$. Following the syntax of $\rec, \rep$, let $R_2, R_3, \cdots, R_{2T} \in \{0, 1\}^n$ and $C_1,C_2, \cdots, C_{2T+1}$ denote the set of \textit{deduced} $Z$-errors on the physical and ancilla qubits respectively, using the information from the meta-checks $\mathcal{S}^0$.\footnote{Note that $R, C$ define $Z(\beta)$, in the notation of \cref{section:rec/rep}.} For $E\in \{0, 1\}^n$, $\mathsf{Syn}_X(E)\in \{0, 1\}^{m_x}$ is the $X$ syndrome vector associated to $E$ (when $E$ is interpreted as a Pauli $Z$ error), and $\mathsf{Syn}_Z(E)$ is defined analogously.

By definition, $R, C$ and $P, B$ are both consistent with syndromes inferred from $\mathcal{S}^0$, which implies the relation
\begin{equation}
\label{eq:meta-check-x}
    B_{2t}\oplus B_{2t+2}\oplus \mathsf{Syn}_X(P_{2t+1}) = C_{2t}\oplus C_{2t+2}\oplus \mathsf{Syn}_X(R_{2t+1}), \quad \forall t=1,2,\dots,T-1,
\end{equation}

\noindent and similarly for the $Z$ syndrome on even layers:
\begin{equation}
\label{eq:meta-check-z}
    B_{2t-1}\oplus B_{2t+1}\oplus \mathsf{Syn}_Z(P_{2t}) = C_{2t-1}\oplus C_{2t+1}\oplus \mathsf{Syn}_Z(R_{2t}), \quad \forall t=1,2,\dots,T.
\end{equation}

We represent the decoding process on a pair of new graphs, coined the ``syndrome adjacency graphs",  following \cite{Gottesman2013FaulttolerantQC}'s ideas. The $X$ (resp, $Z$) syndrome adjacency graph is defined on $T+1$ layers, where in each layer we place $n$ nodes. The 1st and $T+1$st layers are referred to as the ``boundary'' nodes. For each timestep $t\in [T]$ and X parity check $c\in [m_x]$, we also place a node $(c, t)$ \textit{between} code layers $t$ and $t+1$. Intuitively, the nodes in the X syndrome adjacency graph are in bijection with the nodes in the Z layers of the ATG.

In turn, the edges in this graph represent the connectivity of the stabilizers of the ATG: we connect any two nodes if they both are both acted on by a X (resp. Z) meta-checks or X (resp. Z)  boundary code stabilizers (See \cref{fig:xsag}, \cref{fig:zsag}). That is, We connect code nodes $(i, t)$ and $(j, t)$ if $i, j\sim c$ are both in the support of some check $c\in [m_x]$, we connect $(c, t)$ to both $(i, t)$ and $(i, t+1)$, and finally we connect $(c, t)$ to $(c, t+1)$. If the code is $\ell$-LDPC, the degree of this graph is $z\leq \ell(\ell - 1) + 2\ell \leq \ell(\ell+1).$



 To represent the decoding process on these new graphs, we will mark the vertices in which decoding has failed: for $t\in \{1, \cdots T-1\}$ we paint a code vertex $(i, t)$ in the Bulk of the X syndrome adjacency graph iff the physical error differs from the inferred error on the associated vertex of the ATG: $P_{(i, 2t+1)}\neq R_{(i, 2t+1)}$; similarly, we paint a check vertex $(c, t)$ of the syndrome adjacency graph iff $B_{(c, 2t)}\neq C_{(c, 2t)}$. Finally, we paint a vertex on the boundary of the X syndrome adjacency graph iff $\rep_{X}(P)$ is non-zero on the associated qubit of $\partial$.

 \subsection{Technical Proofs}

Perhaps the key observation in the approach of \cite{Kovalev2012FaultTO, Gottesman2013FaulttolerantQC} is to decompose the marked vertices in these adjacency graphs into (maximal) connected components\footnote{A connected component $K$ of marked vertices is maximal if there is no marked vertex adjacent to but not in $K$.}. By definition, the total weight of $R+C$ must be less than the total error on the Bulk $P+B$. What is non-trivial is that this is true even within each connected components of marked vertices, as formalized in \cref{lemma:cc_weight} below.

\begin{lemma}
    [Mismatched errors form connected components]\label{lemma:cc_weight} Consider a connected component $K$ within the Bulk of the X syndrome adjacency graph. Then, 
    \begin{align}
    \sum_{t\in [1, T-1]} \big|R_{2t+1}\big|_K+  \sum_{t\in [1, T]}\big|C_{2t}\big|_K \leq  \sum_{t\in [1, T-1]} \big|P_{2t+1}\big|_K+ \sum_{t\in [1, T]} \big|B_{2t}\big|_K,
    \end{align}

    \noindent and similarly for the Z syndrome adjacency graph. 
\end{lemma}

\begin{proof}
    By definition, $R, C$ are chosen to have minimal weight consistent with the syndrome $\mathcal{S}^0$ information. We claim that if on $K$, the weight of $R, C$ are not minimal, then we could swap $R|_K, C|_K\leftrightarrow P|_K, B|_K$ on $K$, and decrease the overall error weight. Indeed, to see that performing such a swap still produces an operator $R', C'$ consistent with the syndrome information, note that (1) all checks contained entirely within $K$ are consistent, since so is $P|_K, B|_K$; (2) on the ``closure" $\bar{K}$ of $K$ ($K$ and its boundary), $R'|_{\bar{K}}, C'|_{\bar{K}} = P|_{\bar{K}}, B|_{\bar{K}}$, since $K$ is a \textit{maximal} connected component.
\end{proof}

The clustering arguments in the next lemmas require a short technical fact.

\begin{fact}
    [\cite{Kovalev2012FaultTO, Gottesman2013FaulttolerantQC}]\label{fact:set_counting} Consider a set $T$ of $t$ nodes in a graph $G$ of degree $\leq z$. The number of sets $S$ of nodes which contain $T$, of total size $s$, and which form a union of connected components in $G$, is $\leq z^{s-t}\cdot 4^s$. 
\end{fact}

The first step in the clustering argument is to reason that there does not exist any connected component (in either X or Z syndrome adjacency graph) which connects the two boundary codes in $\partial$. This is the only location the third dimension/number of layers of the graph state, $T$, appears. Henceforth, we will refer to this non-connected boundary condition as the \textit{clustering condition}, or $\mathsf{CC}_X, \mathsf{CC}_Z$. As we discuss shortly, conditioned on this event, the residual errors can be described as local stochastic errors.

\begin{lemma}
    [The Boundaries aren't Connected]\label{lemma:boundaries_connected} There exists $p_0\in (0, 1)$ s.t. $\forall p<p_0$, the probability there exists a connected component $K$ in the Bulk of the X syndrome adjacency graph which spans the two boundaries is
    \begin{equation}
    1-\mathbb{P}[\mathsf{CC}_X]\equiv  \mathbb{P}_{P\leftarrow N(p)}\bigg[\exists K\text{ which spans }\partial\bigg] \leq m_x\cdot \frac{(p/p_0)^{T/2}}{1-\sqrt{p/p_0}},
    \end{equation}

    \noindent where $p_0\equiv (8z)^{-2}$. The $Z$ clusters are analogous.
\end{lemma}

\begin{proof}
    Let us fix a connected component of size $|K|=s$. We aim to find a lower bound on the number of Bulk Z errors which occur in $K$, as a function of $s$. For this purpose, note that the number of marked vertices satisfies:\footnote{Here, we let $|RP|$ denote the number of locations the vectors $R, P$ differ, or alternatively, the weight of the operator associated to the product of $R, P$ as Pauli Z operators.}
\begin{equation}
    \begin{aligned}
        s&= \sum_{t\in [1, T-1]} \big|R_{2t+1}P_{2t+1}\big|_K + \sum_{t\in [1, T]} \big|C_{2t}B_{2t}\big|_K \\ &\leq  \sum_{t\in [1, T-1]}\left(\big|R_{2t+1}\big|_K+\big|P_{2t+1}\big|_K\right) +  \sum_{t\in [1, T]} \left(\big|C_{2t}\big|_K+\big|B_{2t}\big|_K\right).
    \end{aligned}
\end{equation}
    Next, we leverage the fact that the weight of the inferred error $(R, C)$ is \textit{minimal}, from \cref{lemma:cc_weight}. That is, within any connected component $K$, 
    \begin{equation}
      \sum_{t\in [1, T-1]}\big|R_{2t+1}\big|_K +  \sum_{t\in [T]} \big|C_{2t}\big|_K\leq  \sum_{t\in [1, T-1]}  \big|P_{2t+1}\big|_K+ \sum_{t\in [T]} \big|B_{2t}\big|_K.
    \end{equation}
This implies that
\begin{equation}
    \sum_{t\in [1, T-1]}  \big|P_{2t+1}\big|_K+ \sum_{t\in [T]} \big|B_{2t}\big|_K \geq \frac{s}{2}.
\end{equation}

    \noindent Therefore, there are at least $s/2$ physical errors ($P, B$) in $K$. Finally, if a connected component $K$ spans the boundaries of $\partial$, it must have size $s \geq T$. By a union bound, 
\begin{equation}
    \begin{aligned}
    &\mathbb{P}_{P\leftarrow N(p)}\bigg[\exists K\text{ which spans }\partial\bigg]\\ = & \mathbb{P}_{P\leftarrow N(p)}\bigg[\exists K, \exists i, j \in [m_x] \text{ s.t. }(i, 1), (j, T)\in K\bigg]\\ 
         \leq &\sum_{s\geq T} \bigg(\text{\# Clusters of Size }s \text{ incident on }\partial\bigg)\cdot \bigg(\text{\# Error Patterns}\bigg)\cdot p^{s/2}\\
         \leq & m_x\cdot \sum_{s\geq T} (4z)^s\cdot 2^s\cdot p^{s/2} \leq m_x\frac{(p/p_0)^{T/2}}{1-\sqrt{p/p_0}}
    \end{aligned}
\end{equation}

    \noindent where $p_0\equiv (8z)^{-2}$. In the last inequality, we leveraged \cref{fact:set_counting} and the fact that the number of ways to pick $s/2$ locations out of a set of size $s$ is $\leq 2^s$.    
\end{proof}

We are now in a position to prove the residual errors $\mathsf{Rep}_Z, \mathsf{Rep}_X$ are local stochastic noise, so long as we condition on the disconnected boundaries condition $\mathsf{CC}$ of \cref{lemma:boundaries_connected}. The proof strategy is similar to that above, and that of \cite{Gottesman2013FaulttolerantQC}: We relate the size of the clusters to the number of true physical errors within it, and subsequently union bound over such configurations of clusters. 

To proceed, we need another short lemma on the weight of clusters connected to \textit{only one} of the boundary codes:

\begin{lemma}\label{lemma:cc-on-boundary}
    Consider a connected component $K$ in the X syndrome adjacency graph, incident on only one of the boundary codes. Let $G=\rep_X(P)$ denote the residual $Z$ error on $\partial$. Then, 
    \begin{equation}
        |G|_K \leq \sum_{t\in [1, T-1]} \big|P_{2t+1}\cdot R_{2t+1}\big|_K. 
    \end{equation}
\end{lemma}

\begin{proof}

    The crux of the proof lies in the following claim (which we prove shortly): If the connected component $K$ is incident on only one of the boundaries of $\partial$, then the operators $G|_K\in \mathsf{Pauli}(n)$ and $\prod_{t\in [1,T-1]} P_{2t+1}|_{K}\cdot R_{2t+1}|_{K}\in \mathsf{Pauli}(n)$ have the same $X$ syndrome. This tells us $G|_{K}$ and $\prod_{t\in [1, T-1]} P_{2t+1}|_{K}R_{2t+1}|_{K}$ differ only by a $Z$ stabilizer. However, if $G$ is \textit{minimal}, then $G|_{K}$ must have weight less than that of $\prod_{t\in [1, T-1]} P_{2t+1}|_{K}R_{2t+1}|_{K}$. Otherwise, we could replace $G$ with $G\cdot G|_{K}\cdot  \prod_{t\in [1, T-1]} P_{2t+1}|_{K}R_{2t+1}|_{K}$, and decrease the overall weight of $G$, without changing the syndrome information. Thus, 

    \begin{equation}
         |G|_K \leq \bigg|\prod_{t\in [1, T-1]} P_{2t+1}\big|_{K}R_{2t+1}\big|_{K}\bigg|\leq \sum_{t\in [1, T-1]} \big|P_{2t+1}\cdot R_{2t+1}\big|_K. 
    \end{equation}

    To prove the missing claim, we note two facts. Let us assume, WLOG, that $K$ is incident on the boundary on the first layer. First, recall that the residual error $G|_K = \mathsf{Rep}_X(P)|_K$ is, by definition, the minimum weight $Z$ error consistent with the \textit{residual} $X$ syndrome (restricted to the cluster $K$). By the definition of the Pauli frame in \cref{fig:rec}, we observe that this residual syndrome is simply $B_2|_K\oplus C_2|_K$, the mismatch on the first layer of $X$ checks. 

    Therefore, it remains to show that $\mathsf{Syn}_X(\prod_{t\in [1,T-1]} P_{2t+1}|_{K}\cdot R_{2t+1}|_{K}) = B_2|_K\oplus C_2|_K$. By a telescoping argument using \cref{eq:meta-check-x},
\begin{equation}
    \begin{aligned}
        \mathsf{Syn}_X( \prod_{t\in [1,T-1]} P_{2t+1}|_{K}\cdot R_{2t+1}|_{K})&=\sum_{t\in [1,T-1]} \mathsf{Syn}_X(P_{2t+1}|_{K}\cdot R_{2t+1}|_{K})\\& = \sum_{t\in [1,T-1]} \bigg(B_{2t}|_K\oplus B_{2t+2}|_K\oplus C_{2t}|_K\oplus C_{2t+2}|_K\bigg)\\& = B_2|_K\oplus C_2|_K \oplus B_{2t^*}|_K\oplus C_{2t^*}|_K,
    \end{aligned}
\end{equation}
   \noindent where we assume the cluster $K$ is entirely contained within layers $[1, t^*<T]$. However, note that the last layer of $K$ must be a ``code" layer, i.e. we must have $B_{2t^*}|_K\oplus C_{2t^*}|_K = 0$. As otherwise, by \cref{eq:meta-check-x} and the meta-check connectivity, we must have at least one connected node at some layer $>2t^*$, a contradiction to $K$ being contained within $[1, t^*<T]$. 
\end{proof}

\begin{lemma}
    [The Residual Error is Stochastic] \label{lemma:residual_stochastic} Let $S\subset \partial$ denote a subset of qubits on the boundary of size $|S|=a$. Then there exists $p_1\in (0, 1)$ s.t. $\forall p < p_1$,
    \begin{equation}
         \mathbb{P}_{P\leftarrow N(p)}\bigg[S\subseteq \mathsf{Supp}(\mathsf{Rep}_X(P))\bigg| \mathsf{CC}_X\bigg] \leq \frac{(p/p_1)^{a/2}}{1-(p/p_1)^{1/4}} \cdot \frac{1}{\mathbb{P}[\mathsf{CC}_X]},
    \end{equation}

    \noindent Where $p_1 = (8z)^{-4}$, and $ \mathbb{P}[\mathsf{CC}_X]$ is defined in \cref{lemma:boundaries_connected}. $\mathsf{Rep}_Z(P)$ is analogous. 
\end{lemma}

\begin{proof}
We wish to bound the probability that the residual error $G = \mathsf{Rep}_X(P)$ is supported on a subset $S$ of size $|S|= a$. Let us once again decompose the $X$ syndrome adjacency graph into clusters, and sum up the total size of clusters connected to (and including) $G$, let this size be $s$. Assuming none of these clusters span the two boundary codes, we have from \cref{lemma:cc-on-boundary},
    \begin{equation}
        2a \leq 2\cdot \sum_K |G|_K \leq \sum_K |G|_K +  \sum_{t\in [1, T-1]} \big|P_{2t+1}R_{2t+1}\big|_K \leq s.
    \end{equation}

    \noindent Moreover, since the weight of $R, C$ is minimal, they must have weight less than that of $P, B$ on each cluster (\cref{lemma:cc_weight}):
\begin{equation}
    \begin{aligned}
        s &= \sum_K \bigg( |G|_K+  \sum_{t\in [1, T-1]} \big|R_{2t+1}P_{2t+1}\big|_K +    \sum_{t\in [1, T]} \big|C_{2t}B_{2t}\big|_K \bigg)\\
        &\leq 4\cdot \sum_K    \bigg(\sum_{t\in [1,T-1]} \big|P_{2t+1}\big|_K+\sum_{t\in [1, T]} \big|B_{2t}\big|_K\bigg). 
    \end{aligned}
\end{equation}

    \noindent In other words, there must be at least $s/4$ real errors ``connected" to the residual error $G$. We can now apply a union bound over the cluster configurations:
\begin{equation}
    \begin{aligned}
    &\mathbb{P}_{P\leftarrow N(p)}\bigg[S\subseteq \mathsf{Supp}(\mathsf{Rep}_X(P)) \text{ and }\mathsf{CC}_X\bigg] \\ \leq&  \sum_{s\geq 2a} \bigg( \text{\# Clusters of Size }s \bigg) \cdot \bigg(\text{\# Error Patterns} \bigg) \cdot p^{s/4} \\\leq& \sum_{s\geq 2a} \big( 4^s\cdot z^{s-a}\big)\cdot \big(2^s\big)\cdot p^{s/4} \leq \big( 2^{6} z \sqrt{p}\big)^{a} \cdot \sum_{i\geq 0} (2^3zp^{1/4})^i \leq \frac{(p/p_1)^{a/2}}{1-(p/p_1)^{1/4}}.
    \end{aligned}
\end{equation}

\noindent So long as $p \leq p_1\equiv (8z)^{-4}$. In the above, we leverage \cref{fact:set_counting}. Bayes rule with $\mathbb{P}[\mathsf{CC}_X]$ concludes the proof. The $Z$ errors are analogous.
\end{proof}

It only remains to consider the logical stabilizers of $\bar{\Phi}$. As described in \cref{sec:stabilizers}, $\mathsf{Rep}_{\bar{X}} \in \{\mathbb{I},\bar{Z}_1\} $ quantifies the necessary logical correction operation on $\partial$, to ensure the resulting state is $\bar{\Phi}$. The lemma below stipulates that except with exponentially small probability, this correction operator is simply identity.

\begin{lemma}
    [There are no Logical Errors]\label{lemma:logical_errors} There exists $p_2\in (0, 1)$ s.t. $\forall p < p_2$, the probability the $X$ logical correction $\mathsf{Rep}_{\bar{X}}(P)$ is non-trivial is
    \begin{equation}
        \mathbb{P}_{P\leftarrow N(p)}\bigg[\mathsf{Rep}_{\bar{X}}(P)\neq \mathbb{I}\bigg] \leq 2\cdot n\cdot T \cdot \frac{(p/p_2)^{d/4}}{1-(p/p_2)^{1/4}} \cdot \frac{1}{\mathbb{P}[\mathsf{CC}_X]},
    \end{equation}

    \noindent where $p_2 =  (8z)^{-4}$, and $\mathbb{P}[\mathsf{CC}_X]$ was defined in \cref{lemma:boundaries_connected}. The $Z$ correction $\mathsf{Rep}_{\bar{Z}}(P)$ is analogous.
\end{lemma}

\begin{proof}
    Suppose, after we apply $\mathsf{Rep}_{X}(P) = \mathsf{Rep}_{X}(P) _1\otimes \mathsf{Rep}_{X}(P)_{2T+1}$ to the boundary $\partial$ of the post-measurement state, we were able to perform a perfect (noiseless) syndrome measurement of an encoded $\bar{X}_1\otimes \bar{X}_{2T+1}$ stabilizer. By definition of the associated encoded logical stabilizers from  \cref{sec:stabilizers}, the resulting syndrome outcome $s_{\bar{X}}\in \{0, 1\}$ is
    
    \begin{equation}
        s_{\bar{X}} = \mathsf{Syn}_{\bar{X}_1\bar{X}_{2T+1}}\bigg(\mathsf{Rep}_{X}(P)_1 \otimes \mathsf{Rep}_{X}(P)_{2T+1} \bigotimes_{t\in [1, T-1]} P_{2t+1}\cdot R_{2t+1}\bigg).
    \end{equation}

    \noindent If we let the support of $\bar{X}$ be $\alpha_x\subset [n]$, this syndrome outcome is equal to the parity of the following operator $ Z(\gamma), \gamma\subset [n]$ on $\alpha_x$:
    \begin{equation}
            Z(\gamma) \equiv \mathsf{Rep}_{X}(P)_1 \cdot \mathsf{Rep}_{X}(P)_{2T+1}\cdot  \prod_{t\in [1, T-1]} P_{2t+1}\cdot R_{2t+1}, \quad s_{\bar{X}} = |\gamma\cap \alpha_x| \mod 2
    \end{equation}

    \noindent We claim that this parity is $0$ with high probability, such that no logical correction $\mathsf{Rep}_{\bar{X}}(P)_1$ is needed. To see this, we follow the procedure from the previous proofs, and decompose the X syndrome adjacency graph into clusters. Recall via the proof of \cref{lemma:cc-on-boundary} that each cluster defines a logical (or trivial) operator of $Q$, since

    \begin{equation}
        \mathsf{Syn}_X\bigg( \mathsf{Rep}_{X}(P)|_{K}\cdot \prod_{t\in [1, T-1]} P_{2t+1}|_{K}R_{2t+1}|_{K}\bigg) = 0.
    \end{equation}

    \noindent Moreover, if this product of operators on $K$ is an X stabilizer of Q, then by definition it has $0$ logical $\bar{X}_1\bar{X}_{2T+1}$ syndrome. The only remaining possibility lies in if the cluster $K$ defines a logical operator. If so, we must have that the size of $|K|=s\geq d$, the distance of the code. By a similar argument behind \cref{lemma:cc-on-boundary} and \cref{lemma:cc_weight}, if we condition on $K$ not spanning the boundaries (event $\mathsf{CC}_X$), then $K$ contains at least $\frac{s}{4}$ physical errors.

    Via the percolation-based union bound, the probability there exists any cluster $K$ with non-trivial $\mathsf{Syn}_{\bar{X}_1\bar{X}_{2T+1}}$ syndrome is then
\begin{equation}
    \begin{aligned}
        &\mathbb{P}_{P\leftarrow N(p)}\bigg[ \mathsf{Syn}_{\bar{X}_1\bar{X}_{2T+1}}\bigg(\mathsf{Rep}_{X}(P)_1 \otimes \mathsf{Rep}_{X}(P)_{2T+1}  \bigotimes_{t\in [1, T-1]} P_{2t+1}\cdot R_{2t+1}\bigg) = 1 \text{ and }\mathsf{CC}_X\bigg]\\
        &\leq  \sum_{s\geq d} \bigg( \text{\# Clusters of Size }s \bigg) \cdot \bigg(\text{\# Error Patterns} \bigg) \cdot p^{s/4} \\
        &\leq 2\dot n\cdot T \sum_{s\geq d} \big( 4^s\cdot z^{s}\big)\cdot \big(2^s\big)\cdot p^{s/4} \leq 2\dot n\cdot T  \cdot \frac{(p/p_2)^{d/4}}{1 - (p/p_2)^{1/4}}
    \end{aligned}
\end{equation}

\noindent where $p_2 = (8z)^{-4}$. In the above, we leveraged \cref{fact:set_counting}.

\end{proof}

We can put the above together and conclude the proof of \cref{theorem:main}.

\begin{proof}

    [of \cref{theorem:main}] By combining \cref{lemma:boundaries_connected} and \cref{lemma:logical_errors}, so long as the noise rate $p< p^* = \min(p_0, p_1, p_2)$, no connected component spans the two boundary codes and $\rec$ doesn't impart a logical error on the encoded state with probability all but 
    \begin{equation}
        \mathbb{P}[\text{Decoding Succeeds}] \geq 1- n\cdot T\cdot  \bigg(\frac{p}{p^*}\bigg)^{\Omega(\min(T, d))}.
    \end{equation}

    \noindent Conditioned on this event, via \cref{lemma:residual_stochastic}, the residual error on the boundaries is local stochastic with noise rate $(p/p_1)^{1/2}$.
    
\end{proof}

\section*{Acknowledgements}
We thank Zhiyang He (Sunny) for helpful discussions. We thank Jong Yeon Lee for helpful discussions and for informing us about an upcoming work with Isaac Kim on single-shot error correction and cluster states. This work was done in part while the authors were visiting the Simons Institute for the Theory of Computing. 
T.B.~acknowledges support by the National Science Foundation under Grant No. DGE 2146752.
Y.L.~is supported by DOE Grant No. DE-SC0024124, NSF Grant No. 2311733, and DOE Quantum Systems Accelerator.

\printbibliography

\appendix

\section{Factorization of the Stabilizers of the \texorpdfstring{$\ket{\atg}$}{|ATG>}}
\label{section:factorize}

We dedicate this section to the proofs that the stabilizers defined in \cref{sec:stabilizers} factorize into X Pauli operators on the Bulk $\mathcal{B}$. To simplify notation, let $N_u = \{ v: (u, v)\in E\}$ denote the neighborhood of a vertex $u\in V$ in the graph $G$. 

\begin{lemma}
    [The Meta-Checks $\mathcal{S}^0$]\label{lemma:meta-checks} For each even layer $t$ and $c\in [m_z]$, the associated $Z$ meta-check factorizes into a product of X Pauli operators on $\mathcal{B}$:

\begin{equation}
    G_{(c, t-1)}\cdot G_{(c, t+1)}\cdot \prod_{i\sim c} G_{(i, t)} = X_{(c, t-1)}\otimes X_{(c, t+1)}\bigotimes_{i\sim c} X_{(i, t)} 
\end{equation}

    \noindent And similarly for the $X$ meta-checks. 
\end{lemma}

\begin{proof}
    For any meta-check centered around the ``meta-vertex" $(c, t)$, the product of graph state stabilizers around it can be written in terms of its ``neighborhood of neighborhoods". The neighborhood $N_{(c, t)}$ of $(c, t)$ consists of the check qubits $(c, t\pm 1)$, and the code qubits $(i, t)$ s.t. $H^Z_{c, i} = 1$ (i.e. $i\sim c$). The ``neighborhood of neighborhoods" is then the multi-set (a set with repetitions) $\cup_{u\in N_{(c, t)}}\cup_{v\in N_u}$, which allows us to write the product of stabilizers above as:

\begin{equation}
    G_{(c, t-1)}\cdot G_{(c, t+1)}\cdot \prod_{i\sim c} G_{(i, t)} = X_{(c, t-1)}\otimes X_{(c, t+1)}\bigotimes_{i\sim c} X_{(i, t)} \prod_{u\in N_{(c, t)}} \prod_{v\in N_u} Z_v
\end{equation}

    We claim that each qubit in $\cup_{u\in N_{(c, t)}}\cup_{v\in N_u}$ appears an even number of times in this multi-set, such that the product of $Z$ Paulis above cancels. In fact, there are only two cases we need to consider: First, consider the code qubits $(i, t\pm 1)$ in layers above and below $t$. Each such node is in both $N_{(c, t\pm 1)}$ and in $N_{(i, t)}$, thus counted twice. 
    
    The challenge lies in the $X$ check qubits at layer $t$. For any $d\in [m_x]$, the check qubit $(d, t)$ lies in the ``neighborhood of neighborhoods" of $(c, t)$ through a code qubit $(i, t)$ iff $H^Z_{c, i}\cdot H^X_{d, i} = 1$. Thereby, the parity of the number of appearances is

    \begin{equation}
        \sum_i H^Z_{c, i}\cdot H^X_{d, i} = \big(H^Z (H^X)^T\big)_{c, d} = 0
    \end{equation}
    
    \noindent Since $(H^X, H^Z)$ defines a CSS code.
\end{proof}

The stabilizers in $\mathcal{S}^1$ arise in two types. Recall that $\ket{\Bar{\Phi}}$ consists of an encoded maximally entangled state accross two copies of the LDPC code $Q$. Then, within $\mathcal{S}^1$ there will be stabilizers of the individual boundary codes, and encoded stabilizers of the Bell state.

\begin{lemma}
[The Stabilizers of the Boundary Codes] Each $X$-type or $Z$-type stabilizer $S_\partial$ of the two boundary LDPC codes $Q$ can be written in the form $X(\alpha) \otimes S_\partial\in \mathcal{S}^1$ via a product of graph state stabilizers. Explicitly,
\begin{enumerate}
    \item For each $Z$-type stabilizer $c\in [m_z]$ of $Q$, there exists graph state stabilizers $\in \mathcal{S}^1$, satisfying the decomposition

    \begin{equation}
        G_{(c, 1)} = X_{(c, 1)} \bigotimes_{i\sim c} Z_{(i, 1)}, \quad G_{(c, 2T+1)} = X_{(c, 2T+1)} \bigotimes_{i\sim c} Z_{(i, 2T+1)}
    \end{equation}
    
    \item For each $X$-type stabilizer  $c\in [m_x]$ of $Q$, there exists products of graph state stabilizers $\in \mathcal{S}^1$, satisfying the decomposition
    \begin{gather}
        G_{(c, 2)} \cdot \prod_{i\sim c} G_{(i, 1)} =  X_{(c, 2)} \bigotimes_{i\sim c} X_{(i, 1)}, \quad  G_{(c, 2T)} \cdot \prod_{i\sim c} G_{(i, 2T+1)} =  X_{(c, 2T)} \bigotimes_{i\sim c} X_{(i, 2T+1)}.
    \end{gather}
\end{enumerate}
\end{lemma}

We remark that both of the decompositions above are simply applying an $X$ or $Z$ stabilizer of the LDPC code on either the first or last layer of $G$.  

\begin{proof}
    Case 1 is rather straightforward, as the graph state stabilizer $S_{(c, 1)}$ (resp, $2T+1$) is precisely applying a Z stabilizer of $Q$ on the associated boundary code. Case 2 is more subtle. However, we can follow the reasoning in \cref{lemma:meta-checks}: to show that the $Z$ Pauli's arising from the graph state stabilizers cancel, it suffices to show they are counted an even number of times from the neighborhoods of $(c, 2)$ and $(i, 1)$ where $H_{c, i}^X=1$. Indeed, the code qubits $(i, 2)$ are counted exactly twice (due to the vertical connections), and the Z-type check qubits $(d, 1)$ are counted an even number of times since $Q$ is a CSS code.  
\end{proof}

The encoded Bell pairs are stabilized by products $\Bar{X}_1\otimes \Bar{X}_{2T+1}, \Bar{Z}_1\otimes \Bar{Z}_{2T+1}$ of logical operators. The Lemma below shows how to construct these operators using graph state stabilizers in $G$, and only $X$ operators on the Bulk.

\begin{lemma}
[The Encoded Stabilizers of $\Phi$] Every encoded stabilizer $S$ of $\Phi$ can be written as a product of graph state stabilizers $S_\partial \otimes X(\alpha)\in \mathcal{S}^1$, for some subset $\alpha\subset \mathcal{B}$. Explicitly,
\begin{enumerate}
    \item Let $\alpha_x\subset [n]$ denote the support of a logical $\Bar{X}$ on $Q$. Then, 
    
    \begin{equation}
\prod_{\substack{i\in \alpha_x \\ t \text{ odd }}} G_{(i, t)} = \Bar{X}_1\otimes \Bar{X}_{2T+1}  \bigotimes_{\substack{i\in \alpha_x \\ t \text{ odd }\in \mathcal{B}}}X_{(i, t)}
\end{equation}
    
    \item Let $\alpha_z\subset [n]$ denote the support of a logical $\Bar{Z}$ on $Q$. Then, 

    \begin{equation}
        \prod_{\substack{i\in \alpha_z \\ t \text{ even }}} G_{(i, t)} = \Bar{Z}_1\otimes \Bar{Z}_{2T+1} \bigotimes_{\substack{i\in \alpha_z \\ t \text{ even }}} X_{(i, t)}
    \end{equation}
\end{enumerate}
\end{lemma}

\begin{proof}
    The argument follows the strategy of the previous two lemmas. The alternating even/odd layers implies that the $Z$ operations arising due to the vertical connections are each counted twice, thus cancelling. The connections to ancilla qubits on each layer cancel due to the CSS condition, as logical operators by definition have support $\alpha_x$ (resp $\alpha_z$) with even overlap with the support of $Z$ (resp $X$) parity checks. 
\end{proof}

\section{From Repeated Measurements to SSSP}
\label{section:mbqc}

A fault-tolerant $\ket{\bar{+}}$ state initialization protocol via $T$ rounds of repeated measurements is comprised of:

    \begin{enumerate}
        \item The initialization of $n$ qubits in the product state $\ket{+}^{\otimes n}$.
        \item $T$ rounds of $Z$ syndrome measurements. For each $Z$ check measurement $c\in [m_z]$, we prepare a $\ket{+}_c$ state and apply $\mathsf{CZ}_{i, c}$ gates between $c$ and each $i\in [n]$ such that $c\sim i$. The ancilla is measured in the $\ket{\pm}$ basis. These alternate with

        \item $T$ rounds of $X$ syndrome measurements. These begin with Hadamard gates $H^{\otimes n}$ on all the code qubits, then for each $X$ check measurement $c\in [m_x]$ we prepare a $\ket{+}_c$ state and apply $\mathsf{CZ}_{i, c}$ gates if $c\sim i$, $i\in [n]$. The ancillas are measured in the $\ket{\pm}$ basis, and finally $H^{\otimes n}$ is applied to the code qubits. 
    \end{enumerate}

We model noise during this process by alternating layers of local stochastic noise $E = (E_1, \cdots, E_T)$ with the syndrome measurement circuit $W$:

\begin{equation}
    \Tilde{W}_E = E_T W_T^X W_T^ZE_{T-1}W_{T-1}^XW_{T-1}^Z\cdots E_1 W_1^X W^Z_1 
\end{equation}

Note that since the syndrome measurement circuits are constant depth Clifford circuits, this is without loss of generality for local stochastic, gate-level Pauli noise. Informally, this process is said to be fault-tolerant if, after collecting the syndrome measurements, one is able to correct the resulting state into the logical code state $\ket{\bar{0}}$, up to a local stochastic error $E'$ drawn from some noise rate $q\in [0, 1]$. Put precisely, 

\begin{definition}\label{def:ft-initialization}
    Fix a noise rate $p\in [0, 1]$. A CSS LDPC code admits a \emph{FT state initialization protocol via $T$ rounds of repeated measurements} against local stochastic noise at rate $p\in [0, 1]$ if there exists a pair of functions $\pframe, \res$: 
\begin{gather}
    \pframe:  \{0, 1\}^{T\cdot (m_x+m_z)}\rightarrow \mathsf{Pauli}([n])\\
    \res: \mathsf{Pauli}(([n]\cup [m_x]\cup [m_z])\times [T]) \rightarrow  \mathsf{Pauli}([n])
\end{gather}
\noindent satisfying, with probability all but $\delta$ over the choice of local stochastic error $E\in \mathsf{Pauli}(([n]\cup [m_x]\cup [m_z])\times [T])$,
\begin{equation}
   (\pframe(s)_{[n]} \otimes \ketbra{\pm_s}) \Tilde{W}_E \ket{0}^{\otimes n}\ket{0}^{\otimes T(m_x+m_z)}    = \gamma_s \cdot \res(E) \ket{\bar{0}} \otimes \ket{\pm_s}.
\end{equation}

\noindent where $\res(E)$ is local stochastic with noise parameter $q$.
\end{definition}

In this section, we show how this encoding process, via measurement-based quantum computation, naturally gives rise to the alternating tanner graph state. What is more, we show that an efficient, fault-tolerant repeated measurement protocol can be transformed (in a black-box manner) into an efficient single-shot state preparation protocol.

\subsection{The Structure of the Cluster State}

To map the repeated measurement circuit to a constant depth quantum circuit + measurements, we follow two rules \cite{jozsa2005introductionmeasurementbasedquantum}. We note that it suffices to specify how to implement Hadamard gates, and two-qubit $\mathsf{CZ}$ gates, as these are the gates needed above. Moreover, since these gates are Clifford gates, the single-qubit measurements required to implement the repeated measurement circuit in MBQC are non-adaptive. 

First, for each Hadamard gate acting on a code qubit $i\in [n]$ in the repeated measurements circuit, we create a copy of $i$ and apply a $\mathsf{CZ}$ gate between the two copies. By measuring the first copy in the $\ket{\pm}$ basis, this serves to ``teleport" the qubit state forward (\cref{fig:hadamard}) \cite{jozsa2005introductionmeasurementbasedquantum} (up to a Pauli $X$). The resulting Pauli $X^m$ acting on the qubit depends on the measurement outcome $m\in\{0, 1\}$. In doing so for each of $n$ code qubits, this will give rise to the layered structure of the alternating tanner graph state.

\begin{figure}
    \centering
    \includegraphics[width=0.4\linewidth]{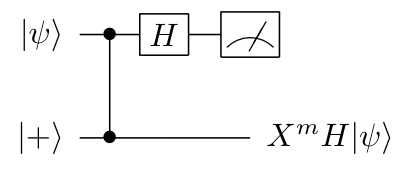}
    \caption{Hadamard gate teleportation via MBQC}
    \label{fig:hadamard}
\end{figure}

Next, for each new ancilla ``syndrome bit" introduced in the repeated measurement circuit, we create a new  qubit initialized in the $\ket{+}$ state in the cluster state; each $\mathsf{CZ}_{i, c}$ gate then acts between the code and ancilla qubits at the corresponding layer/iteration of repeated measurements. The resulting cluster state is then comprised of alternating layers, where there are $n$ code qubits and either $m_x$ or $m_z$ ancilla qubits on each layer. Their connectivity alternates the $X$ and $Z$ tanner graphs within each plane, as well as vertical connections between copies of the same qubit. 

\subsection{The Fault Tolerance Reduction}

It only remains to show now that this transformation preserves fault-tolerance, that is, given a fault-tolerant initialization protocol one can construct a fault-tolerant single-shot state preparation protocol. 

\begin{theorem}\label{theorem:mbqc_mapping}
   Suppose a family $Q$ of CSS quantum LDPC codes admits a \emph{fault-tolerant state initialization protocol via $T$ rounds of repeated measurements} for any local stochastic noise rate $p< p^*\in (0, 1)$, which fails with probability at most $\delta$. Then, it admits a \emph{Single-Shot State Preparation} procedure with $O(T)$ space overhead, for any local stochastic noise rate $p < p^*$, which also fails with probability at most $\delta$.
\end{theorem}

Suppose we run the constant-depth MBQC circuit described above and measure all the qubits in the $\ket{\pm}$ basis, except for the last (the $2T$-th) layer of $n$ code qubits. Let the resulting measured string be $z = z_{\mathsf{Code}}, z_{\mathsf{X}}, z_{\mathsf{Z}}\in \{0, 1\}^{(2T-1)\cdot n}\times \{0, 1\}^{T\cdot m_x}\times \{0, 1\}^{T\cdot m_z}$. At a high level, given a CSS code admitting a state preparation protocol under a pair of frame and residual error functions $(\pframe, \res)$ and the string $z$, our goal will be to construct a set of syndromes $s_{\mathsf{X}}, s_{\mathsf{Z}}\in \{0, 1\}^{T\cdot m_x}\times \{0, 1\}^{T\cdot m_z}$ drawn from the same distribution as in the repeated measurements state preparation protocol.

Before presenting a formal proof, let us first consider the system above in the absence of noise. Note that the first layer of Z syndrome measurements in both the repeated measurements protocol, and the cluster state, are identically distributed, and thus we can pick $z_{\mathsf{Z}}^1=s^1_{\mathsf{Z}}$. Let the resulting state be $\ket{\psi}$. The subsequent vertical $\mathsf{CZ}$ operations serve to teleport $\ket{\psi}$ further, up to a Pauli $X$ or $Z$ correction, dependent of the parity of the layer (\cref{fig:hadamard}):
\begin{equation}
    \label{eq:pauli-shift}
    \ket{\psi} \rightarrow X^{z_{\mathsf{Code}}^1}\ket{\psi}\rightarrow Z^{z_{\mathsf{Code}}^2}X^{z_{\mathsf{Code}}^1}\ket{\psi} \rightarrow Z^{z_{\mathsf{Code}}^2}X^{z_{\mathsf{Code}}^1+z_{\mathsf{Code}}^3}\ket{\psi} \cdots \rightarrow Z^{\sum_i z_{\mathsf{Code}}^{2i}}X^{\sum_i z_{\mathsf{Code}}^{2i-1}}\ket{\psi}
\end{equation}

We account for the effect of these Pauli shifts by subtracting their cumulative syndromes from $z_{\mathsf{Z}}^j, z_{\mathsf{X}}^j$, allowing us to infer $s_{\mathsf{Z}}^j, s_{\mathsf{X}}^j$. Note that in the absence of noise, this simply results in $s_{\mathsf{Z}}^j = s_{\mathsf{Z}}^1$ and $s_{\mathsf{X}}^j=s_{\mathsf{X}}^1=0$. Next, we incorporate the effect of noise. 

\begin{proof}

[of \cref{{theorem:mbqc_mapping}}] Following \cref{def:sssp}, let $F \in \mathsf{Pauli}(([n]^2\cup [m_x]\cup [m_z])\times [T])$ denote a local stochastic error of rate $p$ on the cluster state. Since all but the last $[n]$ code qubits are measured in the $\ket{\pm}$ basis, we can again restrict our attention to Pauli $Z$ errors, which in turn are equivalent to bit-flip errors on the measured string $z=z_{\mathsf{Code}}, z_{\mathsf{X}}, z_{\mathsf{Z}}$. Let $f =f_{\mathsf{Code}}, f_{\mathsf{X}}, f_{\mathsf{Z}} \in \{0, 1\}^{(2T-1)\cdot n}\times \{0, 1\}^{T\cdot m_x}\times \{0, 1\}^{T\cdot m_z}$ denote the bit-flip error on those first $2T-1$ layers. 

Following \cref{eq:pauli-shift} and the discussion above, there is a natural way to construct the (noisy) syndromes and the Pauli error required to invoke the frame and residual error functions $(\pframe, \res)$. First, we identify $s_{\mathsf{X}}, s_{\mathsf{Z}}\in \{0, 1\}^{T\cdot m_x}\times \{0, 1\}^{T\cdot m_z}$ by associating $s_{\mathsf{Z}}^1 = z_{\mathsf{Z}}^1$, and then subtracting the syndrome associated to the Pauli shifts:
\begin{gather}
    s_{\mathsf{Z}}^j = z_{\mathsf{Z}}^j + \mathsf{Syn}_Z \big(\sum_i^j z_{\mathsf{Code}}^{2i-1}\big) \\ 
     s_{\mathsf{X}}^j = z_{\mathsf{X}}^j + \mathsf{Syn}_X \big(\sum_i^j z_{\mathsf{Code}}^{2i}\big)
\end{gather}

\noindent and moreover, associate $f_{\mathsf{X}}, f_{\mathsf{Z}}$ with syndrome measurement errors in the repeated measurement circuit, and define $X^{f_{\mathsf{Code}}^{2i-1}}Z^{f_{\mathsf{Code}}^{2i}}$ to be the Pauli error on the $[n]$ code qubits during the $i$th iteration $W_i$ of the repeated measurement circuit. This suffices since the measured strings $z$ satisfy a given redundancy, even in the presence of noise:
\begin{gather}
    \big(z_{\mathsf{Z}}^j\oplus f_{\mathsf{Z}}^j\big) \oplus \big(z_{\mathsf{Z}}^1\oplus f_{\mathsf{Z}}^1\big) = \mathsf{Syn}_Z \bigg(\sum_i^j z_{\mathsf{Code}}^{2i-1} \oplus f_{\mathsf{Code}}^{2i-1} \bigg) \\
    \big(z_{\mathsf{X}}^j\oplus f_{\mathsf{X}}^j\big) \oplus \big(z_{\mathsf{X}}^2\oplus f_{\mathsf{X}}^2\big) = \mathsf{Syn}_X \bigg(\sum_i^j z_{\mathsf{Code}}^{2i} \oplus f_{\mathsf{Code}}^{2i} \bigg),
\end{gather}

\noindent and therefore the constructed (noisy) syndromes satisfy the same recurrence as if they were drawn from the repeated measurements circuit:

\begin{equation}
    s_{\mathsf{Z}}^j \oplus f_{\mathsf{Z}}^j =  s_{\mathsf{Z}}^{j-1} \oplus f_{\mathsf{Z}}^{j-1} \oplus \mathsf{Syn}_Z\bigg(f_{\mathsf{Code}}^{2j-1}\bigg).
\end{equation}
    
\end{proof}

\section{Single-shot preparation of GHZ States}
\label{section:ghz}

We dedicate this section to the fault-tolerant single-shot preparation of $\ghz$ states. Informally, we prove that if we measure the $\ket{\atg}$ in all but a collection of odd layers $j\cdot \Delta+1, j\in \{0, \cdots k-1\}$ and apply a feed-forward Pauli correction, then the resulting state is an encoded $\ghz$ state $\ket{\overline{\ghz}_m}\propto \ket{\bar{0}}^{\otimes m}+\ket{\bar{1}}^{\otimes m}$. In the presence of local stochastic noise below a threshold rate (determined by the connectivity of the qLDPC code), we can further return the encoded $\ghz$ state up to a residual local stochastic error on each of its ``surfaces". 

\begin{theorem}\label{theorem:ghz}
    Fix integers $T, m\geq 1$. Let $Q$ be any $[[n, k, d]]$ CSS code, which is $\ell$-LDPC. Then, $Q$ admits a single-shot state preparation procedure for the encoded state $\ket{\overline{\ghz^{\otimes k}_m}}$ using a circuit $W$ on $O(\ell\cdot n\cdot T)$ qubits and of depth $O(\ell^2)$, satisfying the following guarantees:

    \begin{enumerate}
        \item There exists a constant $p^*(\ell) \in (0, 1)$, s.t. if $W$ is subject to local stochastic noise of rate $p<p^*$, the state preparation procedure succeeds with probability at least  $1 - n\cdot T\cdot \big(\frac{p}{p^*}\big)^{\Omega(\min(T\cdot m^{-1}, d))}$.
        \item Conditioned on this event, the resulting state is subject to local stochastic noise of rate $O(p^{1/2})$.
    \end{enumerate}
\end{theorem}

In the subsequent subsections, we show that a careful measurement pattern on $\ket{\atg}$ gives rise to the stabilizers of $\ket{\overline{\ghz}_m}$ in the post-measurement state. We present a natural information-theoretic scheme to infer a Pauli frame correction, following \cref{section:rec/rep}. Then, we sketch a proof of fault-tolerance by highlighting the main modifications to the clustering proof of \cref{section:FT-clustering}.

\subsection{The Measurement Pattern}

Starting from the $\ket{\atg}$, suppose we measure (in the $\ket{\pm}$ basis) all qubits except for those $n$ code qubits located on \textit{odd} layers $i_1, i_2, \cdots i_k \in [2T+1]$, where $i_1=1$ and $i_k = 2T+1$. We first claim that the resulting state is, up to a Pauli correction, an encoded $k$-partite GHZ state. To do so, it suffices to prove the stabilizers of the post-measurement state correspond to those of an encoded GHZ state.

\begin{claim}
    The state after measuring $\ket{\atg}$ (in the $\ket{\pm}$ basis), on all but the $n$ code qubits on odd layers $i_1, \cdots i_k\in 2T+1$, is stabilized by

    \begin{enumerate}
    \item The $X$ and $Z$ parity checks of the code, on each layer $i_j$, $j\in [k]$.
    \item A global logical tensor product of $X$'s: $\otimes_j^k \bar{X}_{i_j}$, and logical pairwise tensor product of $Z$'s: $\bar{Z}_{i_j}\otimes \bar{Z}_{i_{j+1}}$, $j\in [k]$.
\end{enumerate}

\noindent and therefore is a logical GHZ state, up to a Pauli correction. 

\end{claim}

The first set of stabilizers above ensures that each layer of the post-measurement state is encoded into the LDPC code; the second ensures that the resulting state is stabilized by the encoded stabilizers of the $k$-partite GHZ state. 

\begin{proof}
    Following \cref{sec:stabilizers}, it suffices to show that each of the operators above can be generated through a product of graph state stabilizers, up to a Pauli $X$ operator on the measured qubits. The proof below is a simple generalization of the Bell state case. We first show that the stabilizers of the code at each layer $i_j$ can be written as a product of graph state stabilizers ``in-plane":
    \begin{equation}
        G_{(c, i_j)} = X_{(c, i_j)} \bigotimes_{l\sim c} Z_{(l, i_j)} \text{ and } 
    G_{(c, i_j-1)}\cdot G_{(c, i_j+1)}\cdot \prod_{l\sim c} G_{(l, i_j)} = X_{(c, i_j-1)}\otimes X_{(c, i_j+1)}\bigotimes_{l\sim c} X_{(l, i_j)} 
\end{equation}

This ensures the set (1) of stabilizers above. Next, in analog to the encoded Bell state of \cref{sec:stabilizers}, the tensor product of logical X operators can be generated via

    \begin{equation}
\prod_{\substack{i\in \alpha_x \\ t \text{ odd }}} G_{(i, t)} = \bigotimes_{\substack{i\in \alpha_x \\ t \text{ odd }}}X_{(i, t)} = \bigotimes_j^k \bar{X}_{i_j} \bigotimes_{\substack{i\in \alpha_x \\ t \text{ odd }\setminus \{i_1, \cdots, i_k\}}}X_{(i, t)},
\end{equation}

 \noindent where $\alpha_x\subset [n]$ denotes the support of a logical $\Bar{X}$ on the code $Q$. Finally, the pairwise tensor product of logical $Z$ operators, can be generated via 

    \begin{equation}
        \prod_{\substack{i\in \alpha_z \\ i_{j}<t<i_{j+1} \text{ even }}} G_{(i, t)} = \Bar{Z}_{i_j}\otimes \Bar{Z}_{i_{j+1}} \bigotimes_{\substack{i\in \alpha_z \\ i_{j}<t<i_{j+1} \text{ even }}} X_{(i, t)}, j\in [k]
    \end{equation}

     \noindent where $\alpha_z\subset [n]$ denotes the support of a logical $\Bar{Z}$ on $Q$. 
\end{proof}

\subsection{Fault Tolerance}

To infer the relevant Pauli frame correction, we proceed similarly to \cref{section:rec/rep}. We first leverage the meta-check stabilizers to infer the errors on the bulk using a minimum weight decoder. Using this inferred error, we make a guess for the residual syndrome on the layers $i_1, \cdots, i_k$. To conclude, we infer a Pauli frame from this residual syndrome. The fault tolerance of this scheme follows the proof of \cref{section:FT-clustering}, however, with a subtle modification to the connectivity of the $Z$ syndrome adjacency graph. The $X$ graph is unmodified. For conciseness, we focus on this modification to the $Z$ graph, and simply highlight the main changes to the rest of the proof. 

The subtlety lies in that under the definition of the Z syndrome adjacency graph in \cref{section:FT-clustering} (based on even code layers), the $i_j$-th layer (which is odd) of qubits isn't contained in the graph. Recall that is the residual $Z$ parity check syndromes lying ``in-plane" at layer $i_j$ which define the residual $X$ error $G_{i_j}$ on that instance of the code. While in the Bell-state case a clear solution was to connect the two boundary $\partial$ layers to the endpoints of the graph, in the GHZ case an elegant solution isn't immediately apparent.\footnote{Furthermore, it is apriori unclear how the correct clustering argument should connect clusters to the residual error.} To adjust, we add to the Z syndrome adjacency graph $k$ copies of $n$ code vertices, connected to the Z checks lying at layers $\{i_j\}_{i\in [k]}$ (See \cref{fig:ghzsag}). We refer to the boundary of this new graph to be the $k\cdot n$ additional connected vertices, and the Bulk to be the remaining vertices.

Following \cref{section:FT-clustering}, we proceed by considering the Bulk of this syndrome adjacency graph, and claim that the connected components of clusters of mismatch locations don't span between the un-measured surfaces (cf. \cref{lemma:boundaries_connected}). So long as the distance between any two adjacent layers is $|i_j - i_{j+1}| \geq \Delta$, and the noise rate is sufficiently below threshold $p\leq p_0$, each cluster is constrained to lie on at most one boundary layer $i_j$ with probability all but exponentially suppressed in $(p/p_0)^{\Delta}$. 

Conditioned on this event, one can now relate the weight of the residual error $G_{i_j}\big|_K$ restricted to a cluster $K$ to the weight of $K$ on either side of $i_j$, and in turn (via \cref{lemma:cc_weight}), the total weight of $K$. Via the percolation argument, the residual errors are local stochastic (following \cref{lemma:residual_stochastic}) and can't lead to logical $\bar{Z}_{i_j}\otimes \bar{Z}_{i_{j+1}}$ errors (following \cref{lemma:logical_errors}). 

\begin{figure}
    \centering
    \includegraphics[width=0.5\linewidth]{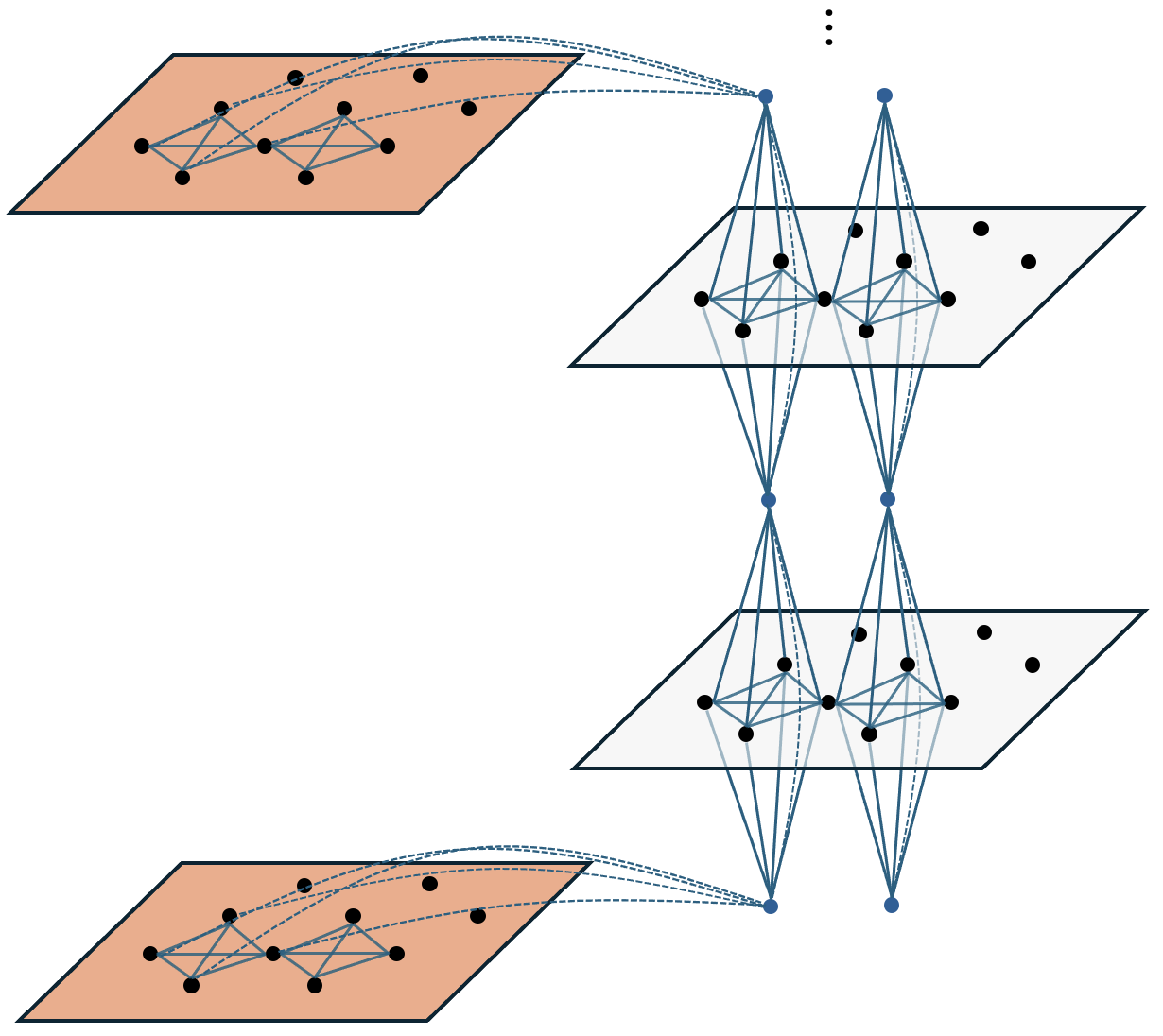}
    \caption{The modified Z syndrome adjacency graph, where the unmeasured odd code layers are connected horizontally to the Z checks on said layers. Pictured are unmeasured nodes corresponding to layers $1, 5$ of the $\atg$ (orange) and measured layers $2, 4$. Not all connections are pictured.}
    \label{fig:ghzsag}
\end{figure}

\section{Preparing CSS Stabilizer States}
\label{sec:cssstateprep}

Suppose $S_Z\subset \{\mathbb{I}, Z\}^{\otimes m}$ and $S_X\subset \{\mathbb{I}, X\}^{\otimes m}$ define a complete commuting set of Pauli operators on $m$ qubits, and assume the weight and degree of each stabilizer in $S_Z$ is upper bounded by a constant $\ell$. Let $m_X = |S_X|$ and $m_Z = |S_Z|$, such that $m = m_X+m_Z$. Let $\ket{\psi}$ be the unique $m$ qubit state which is the $+1$ eigenstate of $S_Z, S_X$. We require a short fact on the preparation of $\ket{\psi}$ using a low depth circuit of $\cnot$ gates with ancillas.

\begin{fact}\label{fact:steanes}
    $\ket{\psi}$ can be prepared by a constant depth circuit with a single round of measurement and Pauli feedforward as follows. 

    \begin{enumerate}
        \item Initialize $m$ data qubits in the $\ket{+}^{\otimes m}$ state, and $m_Z$ ancillas initialized in the $\ket{0}^{\otimes m_Z}$ state.
        \item Apply a $\cnot$ gate between data qubit $i\in [m]$ and ancilla qubit $c\in [m_Z]$ if the $c$-th stabilizer in $S_Z$ is supported on $i$ (the data qubit is control and ancilla qubit is target).
        \item Measure all the ancilla qubits in the $Z$ basis.
        \item Apply a Pauli-$X$ correction based on the measurement outcome. 
    \end{enumerate}
\end{fact}

Roughly speaking, this works because the initial state $\ket{+}^{\otimes m}$ already satisfies all $X$ stabilizers, and subsequent operations can be viewed as measuring all $Z$ stabilizers and correcting them. We prove single-shot preparation of the encoded state $\ket{\psi}$ using a fault tolerant version of \cref{fact:steanes}: step 1 uses single-shot state preparation via the ATG, step 2 uses transversal CNOT gates of CSS codes. Crucially, the Pauli frame in step 1 is propagated through the CNOT gates and corrected at the end.

\begin{theorem}
    Fix an integer $T$. There exists a constant threshold noise rate $p^*$, below which there exists a fault-tolerant, single shot logical state preparation procedure for the encoded $\ket{\psi}^{\otimes k}$ state ($\ket{\psi}$ is a $m$-qubit state) in $m$ copies of an arbitrary $[[n, k, d]]$ CSS LDPC code, using $O(n\cdot T)$ ancilla qubits, with success probability at least $1-m\cdot n\cdot T\cdot 2^{-\Omega(\min(d, T))}$
\end{theorem}

\begin{proof}
    We present a proof by construction. First, $m+m_Z$ copies of the $\ket{\atg}$ are used to implicily prepare the relevant encoded data and ancilla $\ket{\bar{+}}^{\otimes m}$, $\ket{\bar{0}}^{\otimes m_Z}$ qubits (i.e. without measuring the Bulk qubits yet). Subsequently, transversal $\cnot$ gates are used to transversally apply the low-depth circuit guaranteed by \cref{fact:steanes}. The depth of the resulting circuit is a constant, simply that of \cref{fact:steanes} plus that of the $\ket{\atg}$ state preparation. 

    The correctness of this scheme follows from the correctness of \cref{fact:steanes}, the correctness of the $\ket{\atg}$ (\cref{theorem:main}), together with the fact that the transversal $\cnot$ gates apply the circuit independently across the $k$ logical qubits of the CSS code. Indeed, let us first consider the absence of noise. After measuring the Bulk qubits of the $\ket{\atg}$ and the logical ancilla qubits, one can compute the desired Pauli correction as follows. 

    \begin{enumerate}
        \item First computes the relevant Pauli frame for the $\ket{\atg}$ states using the Bulk information. 

        \item Since the $\cnot$ gates are Clifford, one can propagate these Pauli frames through the circuit of \cref{fact:steanes}, to identify the relevant correction after the $\cnot$ gates on both encoded data and ancilla blocks.

        \item Using the logical $Z$ measurements outcomes on all the encoded ancilla states (together with the relevant Pauli frames), one can infer the logical measurement outcomes. 

        \item Finally, compute the feed-forward Pauli correction guaranteed by \cref{fact:steanes}, and apply the corresponding logical Pauli's  to obtain $\ket{\bar{\psi}}$.
    \end{enumerate}

    Note that, if we could apply the corrections of step 2 before the logical $Z$ measurements in step 3, in effect we would be left with the output of the circuit of \cref{fact:steanes} encoded across $m+m_Z$ copies of the LDPC code; and thus, one can infer logical measurement outcomes from the same distribution as that which arises from \cref{fact:steanes}. 

    The presence of noise introduces two complications. First, the correctness of the $\ket{\atg}$ ensures, by a union bound, with probability at least $1-m\cdot n\cdot T\cdot 2^{-\Omega(\min(d, T))}$ all the $m+m_Z\leq 2\cdot m$ logical initialization steps are correct (\cref{theorem:main}). If the residual stochastic noise rate is $q$, then the application of the $\leq \ell^2$ layers of transversal $\cnot$ gates amplifies the effective noise rate to $\leq \ell^2\cdot q$. So long as this new noise rate is below the threshold for the logical measurement, every logical ancilla measurement succeeds with probability $\geq 1- m\cdot 2^{-\Omega(d)}$. 
\end{proof}


\end{document}